\documentclass[11pt]{article}%
\usepackage{amsmath,amsfonts,amsthm}
\usepackage{amssymb}
\usepackage{youngtab}
\usepackage{amsmath}
\usepackage{soul}
\usepackage[left=1in,right=1in,top=1in,bottom=1in]{geometry}

\usepackage{xcolor} 
\usepackage{stmaryrd} 
\usepackage{booktabs} 
\usepackage{bm}

\newtheorem{thm}{Theorem}[section]
\newtheorem{prop}[thm]{Proposition}
\newtheorem{lem}[thm]{Lemma}
\newtheorem{cor}[thm]{Corollary}

\newtheorem{clm}[thm]{Claim}

\DeclareMathOperator{\tr}{tr}
\newcommand{\Tr}[1]{\tr\left[#1\right]}

\usepackage{array}

\usepackage{tikz}
\usepackage{algorithm2e}
\theoremstyle{definition}
\newtheorem{definition}[thm]{Definition}

\usepackage{ytableau}
\usepackage{comment}

\usetikzlibrary{decorations.pathmorphing}
\usetikzlibrary{decorations.pathreplacing}
\usetikzlibrary{decorations.markings} 
\usetikzlibrary{fadings} 
\usetikzlibrary{matrix}
\usetikzlibrary{arrows}
\usetikzlibrary{shapes.arrows} 
\usetikzlibrary{mindmap,trees}
\usetikzlibrary{quantikz}

\theoremstyle{remark}

\numberwithin{equation}{section}

\newcommand{\C}{\mathbb{C}}  

\newcommand{\cV}{\mathcal{V}}

\newcommand{\sh}{\mathsf{h}}

\newcommand{\OO}{\mathrm{O}} 

\renewcommand{\>}{\rangle}
\newcommand{\<}{\langle}
\newcommand{\Dom}{\mathcal{F}}

\newcommand{\ffac}[2]{#1^{\underline{#2}}}
\newcommand{\fFac}[2]{(#1)^{\underline{#2}}}

\newcommand{\classes}{\mathcal{C}} 
\newcommand{\type}[1]{{\mathrm{d}(#1)}} 

\newcommand{\len}[1]{{l(#1)}} 

\newcommand{\hooks}[1]{{\mathsf{H}(#1)}}

\newcommand{\Sym}{S} 

\definecolor{DarkGreen}{rgb}{0,0.6,0}

\definecolor{NEWColor}{rgb}{0,0.4,0.6}

\newcommand{\junk}{T}
\newcommand{\junkSet}{\mathcal{T}}
\newcommand{\Adv}{\mathsf{Adv}}

\newcommand{\regI}{\mathsf{I}}
\newcommand{\regO}{\mathsf{O}}
\newcommand{\regT}{\mathsf{T}}
\newcommand{\regW}{\mathsf{W}}

\begin{document}


\title {A Tight Lower Bound For Non-Coherent Index Erasure}
\author{Nathan Lindzey \\ Department of Computer Science \\ 
University of Colorado at Boulder, USA \\ nathan.lindzey@colorado.edu
\and
Ansis Rosmanis \\ Graduate School of Mathematics \\ Nagoya University, Japan \\ ansis.rosmanis@math.nagoya-u.ac.jp
\thanks{The author is supported by JSPS KAKENHI Grant Number JP20H05966 and MEXT Quantum Leap Flagship Program (MEXT Q-LEAP) Grant Number JPMXS0120319794. Part of this work was done while he was a JSPS International Research Fellow supported by the JSPS KAKENHI Grant Number JP19F19079, and when he was at the Centre for Quantum Technologies at the National University of Singapore supported by the Singapore Ministry of Education and the National Research Foundation under grant R-710-000-012-135.}
}

\maketitle

\begin{abstract} 
	The \emph{index erasure problem} is a quantum state generation problem that asks a quantum computer to prepare a uniform superposition over the image of an injective function given by an oracle. We prove a tight $\Omega(\sqrt{n})$ lower bound on the quantum query complexity of the \emph{non-coherent} case of the problem, where, in addition to preparing the required superposition, the algorithm is allowed to leave the ancillary memory in an arbitrary function-dependent state. This resolves an open question of Ambainis et al., who gave a tight bound for the coherent case, the case where the ancillary memory must return to its initial state. 
	
	To prove our main result, we first extend the \emph{automorphism principle} of H\o{}yer et al. to the \emph{general adversary method} of Lee et al. for state generation problems, which allows one to exploit the symmetries of these problems to lower bound their quantum query complexity. Using this method, we establish a strong connection between the quantum query complexity of non-coherent symmetric state generation problems and the \emph{Krein parameters} of an association scheme defined on injective functions. In particular, we use the spherical harmonics a finite symmetric Gelfand pair associated with the space of injective functions to obtain asymptotic bounds on certain Krein parameters, from which the main result follows.
\end{abstract}

\section{Introduction}

For proving lower bounds in the \emph{oracle query model}, one assumes access to an oracle $O_f$ that evaluates a black-box function $f\colon [n]\rightarrow [m]$ on input queries, where $[n]:=\{1,2,\cdots,n\}$ and $[m]:=\{1,2,\cdots,m\}$, and the goal is to prove that any algorithm for solving the computational problem at hand must make a certain number of oracle queries. This principle for proving lower bounds applies to both classical and quantum computation, and in the latter we allow the oracle to be queried in a superposition. 

Quantum query algorithms are known to surpass their classical counterparts for many important classical tasks, such as unstructured search, game tree evaluation, random walks, and others (see~\cite{Montanaro15,Ambainis2018} for recent surveys). Classical tasks aside, one may also be interested in \emph{quantum mechanical tasks}, such as \emph{quantum state generation}. A quantum state generation problem simply asks for a certain quantum state $\ket{\psi_f}$ to be generated on the target register. In this paper, we consider a particular state generation problem known as \textsc{Index Erasure}.

Given an injective function $f\colon[n] \rightarrow [m]$ via a black-box oracle $O_f$,  \textsc{Index Erasure} is the task of preparing the quantum state that is the uniform superposition over the image of $f$, namely,
$$\ket{\psi_f}:=\frac{1}{\sqrt{n}} \sum_{x=1}^n \ket{f(x)}.$$
The name of the problem stems from the fact that a quantum computer can prepare the uniform superposition $\frac{1}{\sqrt{n}} \sum_{x=1}^n \ket{x} \ket{f(x)}$ using a single query to $O_f$, yet the task of ignoring or \emph{erasing} the first register that records the \emph{index} $x$ is non-trivial. Indeed, if one could solve \textsc{Index Erasure} using a poly-logarithmic number of queries, one would obtain a time-efficient algorithm for \textsc{Graph Isomorphism} (see Section~\ref{app:GraphIso}).

The question of the complexity of \textsc{Index Erasure} was first raised by Shi in~\cite{Shi02},
where he already observed that the problem can be solved in $\mathrm{O}(\sqrt{n})$ queries by an algorithm based on Grover's search. In the same paper, Shi also introduced the \textsc{Set Equality} problem, which asks to decide whether two injective functions $f,f'$ given via black-box oracles $O_f,O_{f'}$ have the same image or have disjoint images, given a promise that either is the case. \textsc{Set Equality} can be easily reduced to \textsc{Index Erasure} via the swap test, increasing the number of oracle queries by at most a constant factor; therefore, when Midrij\=anis presented an $\Omega((n/\log n)^{1/5})$ lower bound on the quantum query complexity of \textsc{Set Equality}~\cite{Mid04}, the same lower bound automatically applied to \textsc{Index Erasure}, ruling out the existence of poly-logarithmic query algorithms for these two problems.

Quantum state generation comes in two forms: the \emph{coherent state generation}, where all memory aside from the target state must return to its initial state,
$\ket{\mathbf{0}}:=\ket{0\cdots0}$,  
and the \emph{non-coherent state generation}, where there is no such a requirement, namely,
where the ancillary memory can remain in some function-dependent state $\ket{t_f}$.
Ambainis, Magnin, Roetteler, and Roland devised the \emph{hybrid adversary method} \cite{AMRR11}, which they used to prove a tight $\Omega(\sqrt{n})$ lower bound for \textsc{Index Erasure} in the \emph{coherent} regime, and left the non-coherent case as an open question. Later, the lower bound for \textsc{Set Equality} was improved to a tight $\Omega(n^{1/3})$ \cite{Zha15,BR18}, which in turn led to an improved query lower bound for the non-coherent \textsc{Index Erasure}. 

In this paper, we close the gap for the non-coherent \textsc{Index Erasure} problem, by proving a tight lower bound on its quantum query complexity under the condition that the range of the black-box function $f$ is sufficiently large.  More formally, we show the following.  
\begin{thm}[Main Result]\label{thm:main}
	The bounded-error quantum query complexity of \textsc{Index Erasure} is $\Theta(\sqrt{n})$ in the non-coherent state generation regime, provided that $m\ge n^{3+\epsilon}$ for some $\epsilon > 0$.
\end{thm}
\noindent We outline the proof of Theorem~\ref{thm:main} below.

\subsection{Outline of the Proof of Theorem~\ref{thm:main}}\label{ssec:outline}

The symmetries of \textsc{Index Erasure} are paramount in our proof (see Section~\ref{sec:prelims} and Section~\ref{sec:assoc} for a detailed discussion of these symmetries and any undefined terminology in this outline). The product $S_n\times S_m$ of two symmetric groups acts on a function $f\colon[n]\rightarrow[m]$ as
$(\pi,\rho)\colon f\mapsto \rho *f*\pi^{-1}$, where $(\pi,\rho)\in S_n\times S_m$ and $*$ denotes the composition of functions. This group action on injective functions defines a representation of $S_n\times S_m$. This representation is multiplicity-free, meaning that it contains no more than one instance of any irrep (irreducible representation) of $S_n\times S_m$.
Moreover, it consists of those and only those irreps $\lambda\otimes\lambda'$ where the Young diagram $\lambda\vdash n$ is contained in the Young diagram $\lambda' \vdash m$ and the skew shape $\lambda'/\lambda$ has no more than one cell per column.
Throughout the paper, we often abuse the terminology and we interchangeably use the terms partition $\lambda$ of $n$, denoted $\lambda\vdash n$, the irreducible representation (irrep) corresponding to $\lambda$, and the $n$-cell Young diagram corresponding to $\lambda$.

The following class of irreps plays a distinguished role in our proof. Given $\lambda\vdash n$, we call the irrep $\lambda\otimes\bar\lambda$ where $\bar\lambda\vdash m$ is obtained from $\lambda$ by adding $m-n$ cells to the first row of $\lambda$ a \emph{minimal} irrep. 
In other words, if $\theta \vdash k$ and $\lambda:=(n-k,\theta)\vdash n$, then $(n-k,\theta)\otimes(m-k,\theta)$ is a minimal irrep.
For example, if $m = 12$, $n = 6$, $k = 3$, and $\theta = (2,1)$, then the minimal 
irrep with respect to $\theta$ is
\begin{align}\label{eq:minmax}
\ytableausetup{mathmode,boxsize=1.0em}
\begin{ytableau}
\bullet & \bullet & \bullet  & \empty &  \empty & \empty & \empty & \empty & \empty \\
\bullet  & \bullet   \\ 
\bullet  \\  
\end{ytableau}
\end{align}
where the $\bullet$'s indicate an $S_n$-irrep and the $\Box$'s indicate an $S_m$-irrep. To lower bound the quantum query complexity of the non-coherent \textsc{Index Erasure}, we use essentially the same adversary matrix $\Gamma$ that~\cite{AMRR11} used for the coherent \textsc{Index Erasure}, which is specified through minimal irreps (see Section~\ref{sec:AdvToKrein} for a formal definition of this matrix). 

An adversary matrix is a symmetric real matrix whose rows and columns are labeled by all the functions in the domain of the problem, and it is the central object of most adversary methods. In our case, the adversary matrix acts on the same $\ffac{m}{n}$\,-dimensional space as the representation matrices of $S_n\times S_m$ mentioned above, where $\ffac{m}{n}:=m!/(m-n)!$ is the total number of functions. Similarly to~\cite{AMRR11}, we choose
\[
\Gamma:=\sum_{k=0}^{\sqrt{n}-1}\left(\sqrt{n} - k\right)
\sum_{\theta\vdash k}E_{(n-k,\theta)\otimes(m-k,\theta)},
\]
where $E_{\lambda\otimes\lambda'}$ is the orthogonal projector on the irrep $\lambda\otimes\lambda'$ (note that we have only used projectors on certain minimal irreps to construct $\Gamma$). We also note that the Gram matrices $T_{\lambda\otimes\lambda'}=\ffac{m}{n} {E_{\lambda\otimes\lambda'}}/{d_{\lambda\otimes\lambda'}}$, where $d_{\lambda\otimes\lambda'}:=\Tr{E_{\lambda\otimes\lambda'}}$ is the dimension of $\lambda\otimes\lambda'$, play an important role in our proof.

In order to take advantage of the inherent symmetries of the \textsc{Index Erasure} problem, we first extend the the \emph{automorphism principle} of H\o{}yer, Lee, and {\v S}palek~\cite{HLS07} to the \emph{general adversary method} for state generation and conversion problems~\cite{LMRSS11} (see Corollary~\ref{cor:adv} and Theorem~\ref{thm:auto}). This extension leads us to consider the Gram matrix corresponding to the final state $\ket{\psi_f,t_f}$ of an algorithm run with oracle $O_f$ (assuming no error).
The Gram matrix corresponding to $\ket{\psi_f}$ is
\[
\frac{n}{m}T_{(n)\otimes(m)}+\left(1-\frac{n}{m}\right)T_{(n)\otimes(m-1,1)}=:T^{\odot},
\]
therefore the Gram matrix corresponding to $\ket{\psi_f,t_f}$ is $T^\odot\circ T$, where $T_{f,f'}:=\< t_f| t_{f'}\>$ and $\circ$ denotes the Schur (i.e., entrywise) matrix product.
For the coherent regime lower bound, $\<\mathbf{0}|\mathbf{0}\>=1$ and $T=J=T_{(n)\otimes(m)}$ is the all-ones matrix. For the non-coherent regime, the Gram matrix $T$ can be arbitrary, but one of the consequences of the generalization of the automorphism principle is that it suffices to consider $T$ such that $T_{f,f'}=T_{\sigma(f),\sigma(f')}$ for all functions $f,f'$ and all $\sigma\in S_n\times S_m$.

To prove the $\Omega(\sqrt n)$ lower bound, we must show, for all such Gram matrices $T$, that
\begin{equation}\label{eq:IntroCond0}
\Tr{\Pi_\Gamma\frac{T^\odot\circ T}{\ffac{m}{n}}} = o(1),
\end{equation}
where $\Pi_\Gamma$ is the orthogonal projector on the image of $\Gamma$, and that $\|\Gamma\circ\Delta_x\|=\mathrm{O}(1)$ for all $x\in[n]$, where $\Delta_x$ is the binary matrix with $(\Delta_x)_{f,f'}:=1$ if and only if $f(x)\ne f'(x)$.\footnote{The terms in condition (\ref{eq:IntroCond0}) and similar expressions are written in such a way to emphasize that $\frac{T^\odot\circ T}{\ffac{m}{n}}$ is a density operator.}
Here we only need to prove the former condition because we use essentially the same adversary matrix as \cite{AMRR11}, and the latter condition is shown in their work. On the other hand, showing condition~(\ref{eq:IntroCond0}) was a triviality in~\cite{AMRR11} because $T = J$ in the coherent regime and thus the trace evaluates to $n/m$. Showing that condition (\ref{eq:IntroCond0}) holds is the main technical contribution of this work, and appears to require significantly more algebraic results on the space of injective functions than the main technical result of~\cite{AMRR11}, which we develop in Sections~\ref{sec:prelims} and~\ref{sec:assoc}. 

We now present the three main simplifying steps used to narrow the scope of condition~(\ref{eq:IntroCond0}).
First, we use linearity to show that it suffices to prove
\[
\Tr{\Pi_\Gamma\frac{T_{(n)\otimes(m-1,1)}\circ T_{\lambda\otimes\lambda'}}{\ffac{m}{n}}} = o(1)
\]
for all irreps $\lambda\otimes\lambda'$.
That is, we can restrict our attention from a continuum of choices for $T$ to a finite set $\{T_{\lambda\otimes\lambda'}\}_{\lambda\otimes\lambda'}$ of choices, where we have also used that the term $T_{(n)\otimes(m-1,1)}$ ``dominates'' $T_{(n)\otimes(m)}$ in $T^\odot$ (since we assume that $m \gg n$).

Second, we use the connection between $T_{(n)\otimes(m-1,1)}$ and a specific primitive idempotent of the Johnson (association) scheme to obtain
\[
\Tr{E_{\lambda\otimes\bar\lambda}\frac{T_{(n)\otimes(m-1,1)}\circ T_{\lambda\otimes\lambda'}}{\ffac{m}{n}}} = o(1)
\]
as a sufficient condition, where we have to consider only Young diagrams $\lambda\vdash n$ that have less than $\sqrt{n}$ cells below the first row.

Third, for such $\lambda$, we show that the dimension of $\lambda\otimes\bar\lambda$ is much smaller than the dimension of any other $\lambda\otimes\lambda'$ (thus the nomenclature ``minimal irrep''); therefore, we show it suffices to prove for all $\lambda \vdash n$ that
\begin{equation}\label{eq:IntroCond3}
\Tr{E_{\lambda\otimes\bar\lambda}\frac{T_{(n)\otimes(m-1,1)}\circ T_{\lambda\otimes\bar\lambda}}{\ffac{m}{n}}} = o(1)
\end{equation} 

It is convenient to think of (\ref{eq:IntroCond0}) and its simplifications in terms of the following association scheme (see Section~\ref{sec:assoc} for more details).
For a pair of functions $(f,f')$, consider the orbit $\mathcal{O}_\mu:=\{(\sigma(f),\sigma(f'))\colon \sigma\in S_n\times S_m\}$, and let $A_\mu$ be the binary matrix with $(A_\mu)_{h,h'}=1$ if and only if $(h,h')\in\mathcal{O}_\mu$. Here we use $\mu$ to label distinct orbits and let $\classes_n$ be the set of all of them. The set of matrices $\{A_\mu\colon\mu\in\classes_n\}$ forms a symmetric association scheme, denoted $\mathcal{A}_{n,m}$, which has been called the \emph{injection scheme}~\cite{Munemasa01}.
Note that there is an obvious bijection between injective functions $f\colon [n]\rightarrow [m]$ and $n$-partial permutations of $[m]$ via $f\leftrightarrow(f(1),f(2),\cdots,f(n))$.\footnote{For describing particular injections, we prefer the latter representation as it is a bit more succinct.}

In the terminology of association schemes, the projectors $E_{\lambda\otimes\lambda'}$ are called the \emph{primitive idempotents}, and their entries corresponding to the orbit $\mathcal{O}_\mu$ multiplied by $\ffac{m}{n}$ are called \emph{dual eigenvalues} of the association scheme, which we denote as $q_{\lambda\otimes\lambda'}(\mu)$. The \emph{valency} $v_\mu$ is the size of $\mathcal{O}_\mu$ divided by $\ffac{m}{n}$, thus, in terms of dual eigenvalues, the left hand side of condition (\ref{eq:IntroCond3}) can be written as
\begin{equation}\label{eq:IntroDualEV}
\frac{\sum\nolimits_{\mu\in\classes_n} v_\mu^{}\cdot q_{(n)\otimes(m-1,1)}(\mu)\cdot q_{\lambda\otimes\bar\lambda}^2(\mu)}
{\ffac{m}{n}~d_{(n)\otimes(m-1,1)}d_{\lambda\otimes\bar\lambda}} = o(1).
\end{equation}
Finally, to prove (\ref{eq:IntroDualEV}), we consider the spherical harmonics of a finite symmetric Gelfand pair associated with the space of injective functions along with some estimates of combinatorial coefficients related to the unsigned Stirling numbers of the first kind. 
\subsection{Krein Parameters of the Injection Scheme}

In the context of quantum query complexity, the injection association scheme was already considered in~\cite{RB13}, where a conjecture on its eigenvalues implied tight adversary bounds for the \textsc{Collision} and \textsc{Set Equality} problems. Along these lines, our work shows a connection between quantum query complexity and the \emph{Krein parameters} $q_{i,j}(k)$ of association schemes (see Section~\ref{sec:AdvToKrein} for a formal definition).
Indeed, condition (\ref{eq:IntroCond3}) is equivalent to the conditions
\[
q_{\lambda\otimes\bar\lambda,\,\lambda\otimes\bar\lambda}((n)\otimes(m-1,1))=o(d_{\lambda\otimes\bar\lambda})
\quad\text{and}\quad
q_{\lambda\otimes\bar\lambda,\,(n)\otimes(m-1,1)}(\lambda\otimes\bar\lambda)=o(m)
\]
on the Krein parameters of $\mathcal{A}_{n,m}$, and (\ref{eq:IntroDualEV}) gives an expression of these parameters in terms of dual eigenvalues. 

The Krein parameters of an association scheme are important because they are the \emph{dual structure constants} of its corresponding \emph{Bose-Mesner algebra}. While the structure constants (i.e., intersection numbers) of Bose-Mesner algebras admit an obvious combinatorial meaning, its dual structure constants do not (e.g., they can be irrational) and are difficult to interpret. Indeed, the question of whether or not there exists a ``good" interpretation of these constants has often been asked in algebraic combinatorics, so we find their connection to quantum query complexity to be interesting.

\subsection{Connection to Graph Isomorphism through Set Equality}\label{app:GraphIso}

\newcommand{\nodeN}{k}

Given a graph $G$ of the vertex set $[\nodeN] := \{1,2,\cdots,\nodeN\}$ and a permutation $\pi\in S_\nodeN$, let $\pi \cdot G$ denote the graph obtained by the natural action of $\pi$ on the vertices of $G$. Let us assume that $G$ is rigid, so the ``permuting'' function $f\colon\pi\mapsto\pi\cdot G$ is injective, and let $O_f$ be the oracle evaluating this function.

Now suppose we have two rigid graphs $G_0,G_1$ of the vertex set $[\nodeN]$, and let $f_0,f_1$ be the corresponding permuting functions. If the two graphs are isomorphic, then $\mathrm{im}\,f_0=\mathrm{im}\,f_1$, while, if they are non-isomorphic, then $\mathrm{im}\,f_0\,\cap\,\mathrm{im}\,f_1=\emptyset$. As a result, we can employ a query-optimal algorithm for {\sc{Set Equality}}~\cite{BHT} which performs $\mathrm{O}(\sqrt[3]{\nodeN!})$ queries to oracles $O_{f_0},O_{f_1}$ and tests isomorphism of $G_0$ and $G_1$ without having to look into internal structure of these graphs.

Because of the optimality of the query algorithm for {\sc{Set Equality}}, one may want to say that any algorithm for {\sc{}Graph Isomorphism} that does not employ the internal structure of graphs must perform  $\Omega(\sqrt[3]{\nodeN!})$ queries to oracles $O_{f_0},O_{f_1}$.
However, to formally prove such a statement, one would have to formalize what is meant by ``not employing the internal structure of a graph''. A potential approach to do that would be to encrypt all graphs using a uniformly random injective function $\mathcal{E}$ from all $\nodeN$-vertex graphs to bit-strings of length $\mathsf{const}\cdot \nodeN$ and to provide an algorithm with encryptions $\mathcal{E}(G_0),\mathcal{E}(G_1)$ and an oracle-access to the function $F\colon (\pi,\mathcal{E}(G))\mapsto\mathcal{E}(\pi\cdot G)$.%
\footnote{The constant $\mathsf{const}$ must be at least $2$, but one may wish to choose it larger so that a randomly guessed bit string is unlikely to be an encryption of any graph.} 

We conjecture that, in this setting, testing isomorphism of $G_0$ and $G_1$ requires $\Theta(\sqrt[3]{\nodeN!})$ oracle evaluations of $F$. However, note that, compared to completely random injective functions on $S_\nodeN$, the function $F$ has an additional structure. For example, $F(\tau,F(\tau,\mathcal{E}(G)))=\mathcal{E}(G)$ for any transposition $\tau$. To prove the desired lower bound, one would have to show that no algorithm can take advantage of this additional structure, and such task is beyond the scope of the present work.\\


The {\sc{Set Equality}} problem can be reduced to {\sc{Index Erasure}}.  Let us describe two natural reductions, one that requires {\sc{Index Erasure}} to be coherent and one that permits it to be non-coherent. Here we assume that we are given oracle access to two injective functions $f_0,f_1\colon[n]\rightarrow[m]$.

In the first reduction, one can prepares state
$
( \ket{0} \ket{\psi_{f_0}} \ket{t_0} + \ket{1}\ket{\psi_{f_1}}\ket{t_1})/\sqrt{2} 
$
such that
$
|\psi_{f}\rangle :=  \sum_{x=1}^n |f(x)\rangle / \sqrt{n}
$
using an {\sc{Index Erasure}} algorithm, and then measures the first qubit in its Fourier basis $\{\ket{+},\ket{-}\}$, where $\ket{\pm}=(\ket{0}\pm\ket{1})/\sqrt{2}$. For this reduction to work, the temporary registers must be left in a default state $\ket{t_0}= \ket{t_1} = \ket{\mathbf{0}}$, which requires index erasure to be coherent. 

In the second reduction, one prepares the state 
$\ket{S_n(G_0)} \ket{t_0} \ket{S_n(G_1)}\ket{t_1}$
using two runs of an {\sc{Index Erasure}} algorithm, and then performs the \textsf{SWAP} test on the registers containing $\ket{S_n(G_0)}$ and $\ket{S_n(G_1)}$. Unlike the reduction before, this reduction does not require {\sc{Index Erasure}} to be coherent. That is, the states  $\ket{t_0},\ket{t_1}$ of the ``garbage'' qubits are inconsequential, and the algorithm does not have to clean them up.

Ambainis et al.~\cite{AMRR11} showed that the quantum query complexity of \emph{coherent} {\sc{}Index Erasure} is $\Theta(\sqrt{n})$, therefore strictly separating complexities of {\sc{}Set Equality} and coherent {\sc{}Index Erasure}. Before the present work, there was still a possibility that non-coherent {\sc{}Index Erasure} might be as fast as {\sc{}Set Equality}, but we prove that it is not the case, strictly separating complexities of {\sc{}Set Equality} and non-coherent {\sc{}Index Erasure} as well.

\subsection{Organization of the paper }

The paper is organized as follows. In Section~\ref{sec:QuantumModel}, we present preliminaries on the quantum query model, with emphasis on state generation problems, including \textsc{Index Erasure}, the general adversary method, and the automorphism principle. In Section~\ref{sec:prelims}, we present preliminaries on the representation theory, particularly focusing on the symmetric group and its action on injections. The automorphism principle of the general adversary method requires us to analyze highly symmetric matrices, which are elements of the Bose--Mesner algebra corresponding to the injection scheme. In Section \ref{sec:assoc}, we formally define this association scheme, establishing the labeling of its various parameters and computing some of them, as well as addressing its connection to the Johnson scheme. With this formalism at our disposal, in Section~\ref{sec:AdvToKrein}, we show that the proof for the $\Omega(\sqrt{n})$ lower bound on the quantum query complexity of the non-coherent \textsc{Index Erasure} can be reduced to showing upper bounds on certain Krein parameters of the injection scheme. Finally, we place the required bounds on these Krein parameters in Section~\ref{sec:MainProof}.

\section{Quantum state generation} \label{sec:QuantumModel}

In this paper, we address limitations of quantum query algorithms for solving the \textsc{Index Erasure} problem. We assume that the reader is familiar with foundations of quantum computing (see~\cite{nielsen2000quantum} for an introductory reference), some of which we review here. The basic memory unit of a quantum computer is a qubit, which is a two-dimensional complex Euclidean space $\C[\{0,1\}]$ having \emph{computational} orthonormal basis $\{\ket{0},\ket{1}\}$. Similarly, a $k$-qubit system corresponds to Euclidean space $\C[\{0,1\}^k]$  {with computational basis $\{\ket{b}\colon b\in\{0,1\}^k\}$}. Unit vectors $\ket{\Psi}\in\C[\{0,1\}^k]$ are called (pure) \emph{quantum states} and they represent superpositions over  {various computational basis states}.

Quantum bits are often grouped together in \emph{registers} for the ease of algorithm design and analysis. If $\ket{\psi},\ket{\phi}$ are states of two registers, then the state of the joint system is $\ket{\psi}\otimes\ket{\phi}$. We often shorten the notation $\ket{\psi}\otimes\ket{\phi}$ to $\ket{\psi}\ket{\phi}$ or $\ket{\psi,\phi}$. Due to \emph{entanglement}, it is not always the case that the state of the joint system can be written as a tensor product of states of the individual registers.

Quantum information is processed by unitary transformations, which correspond to square matrices $U$ such that $UU^*=U^*U=I$, and they map quantum states to quantum states. This unitary processing of quantum information implies that any (noiseless) quantum computation is reversible.

\subsection{Quantum query model}

In the oracle model, we are given an access to a black-box oracle $O_f$ that evaluates some unknown function $f\colon[n]\rightarrow[m]$. The goal of a query algorithm is to perform some computational task that depends on $f$, for example, to compute some function of $f$, such as $\textsc{Parity}(f):=f(1)\oplus f(2)\oplus\cdots\oplus f(n)$ when $m=2$. In quantum computing, one can query the oracle in superposition. On the other hand, due to the requirement for reversibility, the oracle is typically designed so that it preserves the input query $x$. Namely, given $\ket{x,y}$ as an input, the oracle $O_f$ outputs $\ket{x,y\oplus f(x)}$ (see Figure~\ref{fig:oracle}). Here and below we may assume $x,y,f(x)$ to be represented in binary. Even if $f$ is injective---as it is for \textsc{Index Erasure}---unless one knows how to compute the inverse of $f$, implementing $\ket{x}\mapsto\ket{f(x)}$ in practice might be much harder than $\ket{x,y}\mapsto\ket{x,y\oplus f(x)}$.

\begin{figure}[h!]
	\centering
\begin{tikzcd}
\lstick{$\ket{x}_{\regI}$\hspace{3.4pt}} & \gate[wires=2]{O_f} & \qw\rstick{$\ket{x}_{\regI}$} \\
\lstick{$\ket{y}_{\regO}$} & & \qw\rstick{$\ket{y\oplus f(x)}_{\regO}$}
\end{tikzcd}
\label{fig:oracle}
\caption{A schematic of a quantum oracle $O_f$. We assume that $y$ and $f(x)$ are encoded in binary, and thus $O_f$ is its own inverse.}\label{fig:oracle}
\end{figure}

A quantum query algorithm with oracle $O_f$ consists of
\begin{itemize}
\item four registers: input and output registers $\regI$ and $\regO$ for accessing the black-box function $f$, the target register $\regT$ for storing the result of the computation, and a workspace register $\regW$;
\item an indexed sequence of unitary transformations $U_0,U_1,\cdots,U_Q$ acting on those four  registers.
\end{itemize}
The quantum query algorithm starts its computation in state $\ket{\mathbf{0}}:=\ket{00\cdots 0}$, and then performs $2Q+1$ unitary operations, alternating between $U_i$, which acts on all the registers, and $O_f$, which acts on registers $\regI\regO$. Thus the final state of the computation is
\[
\ket{\Psi_f}:=U_Q(O_f\otimes I_{\regT\regW}) U_{Q-1}(O_f\otimes I_{\regT\regW})\cdots U_1(O_f\otimes I_{\regT\regW})U_0\ket{\mathbf{0}},
\]
where $I_{\regT\regW}$ is the identity operator on registers $\regT\regW$.
Figure~\ref{fig:1} gives a schematic of a quantum query algorithm.
Note that $Q$ is the number of oracle queries performed by the algorithm, and we also refer to it as the \emph{query complexity of the algorithm}.
\begin{center}
\begin{figure}[h!]
	\centering
\begin{tikzcd}
\lstick{$\ket{\mathbf{0}}_{\regT}$\hspace{2.4pt}} & \gate[wires=4]{U_0} & \qw & \gate[wires=4]{U_1} &\ \ldots\ \qw & \qw & \gate[wires=4]{U_Q} & \qw \slice[style = {line width=0.5pt,dashed,gray!50!black,align=right}]{\ket{\Psi_f}} & \qw  \rstick{\hspace{6pt}$\approx\ket{\psi_f}$}\\ %
\lstick{$\ket{\mathbf{0}}_{\regI}$\hspace{5.8pt}} & & \gate[wires=2]{O_f} & &\ \ldots\ \qw & \gate[wires=2]{O_f} & &\qw & \qw \rstick[wires=3]{$\approx\ket{t_f}$} \\
\lstick{$\ket{\mathbf{0}}_{\regO}\hspace{2pt}$} & & & &\ \ldots\ \qw & & &\qw & \qw \\
\lstick{$\ket{\mathbf{0}}_{\regW}$} & & \qw & &\ \ldots\ \qw & \qw & &\qw & \qw
\end{tikzcd}
\caption{A schematic of a quantum algorithm that uses an oracle $O_f$. The registers labeled $\regT,\regI,\regO,\regW$ are, respectively, the target, input, output, and workspace registers of the algorithm. 
The target register of the final state $\ket{\Psi_f}$ of the algorithm should be in a state close to the target state $\ket{\psi_f}$.
 }
\label{fig:1}
\end{figure}
\end{center}
In this paper we are interested in quantum query algorithms whose goal is to generate a specific $f$-dependent state $\ket{\psi_f}$ by accessing $f$ via $O_f$. We note that this generalizes classical function evaluation by a quantum algorithm, where each $\ket{\psi_f}$ is asked to be a computational basis vector. In the next section we describe two distinct regimes of quantum state generation, as well as why they are exactly the same for classical function evaluation.
 
\subsection{Coherent vs. Non-coherent State Generation}

When we talk about quantum state generation with oracle $O_f$, we implicitly assume the domain $[n]$ and the range $[m]$ of $f$ to be fixed.
A \emph{quantum state generation} problem is thus specified by a subset $\Dom$ of functions in form $f\colon[n]\rightarrow[m]$, which we call the \emph{domain of the problem}, a complex Euclidean space called the \emph{target space}, and, for every $f\in\Dom$, a quantum state $\ket{\psi_f}$ in the target space called the \emph{target state}.

One may consider quantum state generation in two regimes: coherent and non-coherent.
In the \emph{coherent} state generation regime, all the computational memory other than the target register (i.e., registers $\regI\regO\regW$) must be returned to its initial state $\ket{\mathbf{0}}$. Therefore, if one was running an algorithm for a superposition of oracles, the final quantum state would be a superposition of the target states. In contrast, for \emph{non-coherent} state generation, one does not place any requirements on the ancillary memory. More precisely, in the coherent case, for every input $f\in\Dom$ we require that the final state $|\Psi_f\>$ satisfies
\[ \Re\<\psi_f,\mathbf{0}|\Psi_f\> \ge \sqrt{1-\epsilon}, \]
where $|\mathbf{0}\>$ is the initial state of the ancillary registers and a constant $\epsilon\ge 0$ is the desired precision~\cite{LMRSS11}.
We call the minimum among quantum query complexities among quantum query algorithms that achieve this task the \emph{($\epsilon$-error) quantum query complexity of the coherent version of the problem}.
On the other hand, in the non-coherent case, the final state $|\Psi_f\>$ has to satisfy
\[
\left\|(\<\psi_f|\otimes I)|\Psi_f\>\right\|
= \max_{|t_f\>} \Re \<\psi_f,t_f|\Psi_f\>
\ge \sqrt{1-\epsilon},
\]
where the maximum is over unit vectors $|t_f\>$ on the system of registers $\regI\regO\regW$~\cite{LMRSS11}, and we analogously define the quantum query complexity of the non-coherent version of the problem.

It is worth noting that evaluation of classical functions can be considered as a special case of quantum state generation, where one is asked to prepare the computational basis state $|\psi_f\>$. Since quantum mechanics permits cloning of orthogonal states (computational basis states, in this case), there is no difference between coherent and non-coherent function evaluation, if one is willing to tolerate a two-fold increase in query complexity: at the end of a non-coherent computation, one can copy the target register into an additional register, and then run the whole computation in reverse, restoring all but this additional register to their initial state. 

Finally, note that an algorithm for a coherent case of a problem solves its non-coherent case as well. Conversely, a lower bound on the non-coherent version of the problem is a lower bound on the coherent version as well.

\subsection{Index Erasure}

Throughout this work, we let $n$ and $m$ be positive integers such that $n \leq m$. The domain of \textsc{Index Erasure} is the set of all injective functions $f\colon[n]\rightarrow[m]$. These functions are in one-to-one correspondence with $n$-partial permutations of $[m]$ and thus $|\Dom|=\ffac{m}{n}:=m!/(m-n)!$.
\textsc{Index Erasure} is the task of preparing the quantum state that is the uniform superposition
\[  \ket{\psi_f}:=\frac{1}{\sqrt{n}} \sum_{x=1}^n |f(x)\rangle \]
over the image of $f$.
Note that the state
\[  \frac{1}{\sqrt{n}} \sum_{x=1}^n | x \rangle |f(x)\rangle \]
can be prepared using a single query to $O_f$. This would give us the superposition that we seek if we could only ignore or \emph{erase} the first register that records the \emph{index} $x$, which gives the problem its namesake.

The question of the complexity of \textsc{Index Erasure} was first raised by Shi~\cite{Shi02}.
As for the upper bound, there is a simple quantum query algorithm for coherent \textsc{Index Erasure} given access to $O_f$.  Thinking of the injective function $f$ as a database with entries in $[m]$, for any $y$ in the image of $f$ we may use Grover's algorithm with $O_f$ to find the unique index $x$ of $f$ such that $f(x) = y$.  In other words, there is a circuit that sends the superposition
\[  \frac{1}{\sqrt{n}} \sum_{x=1}^n |f(x)\rangle \quad \text{ to } \quad \frac{1}{\sqrt{n}} \sum_{x=1}^n  |x\rangle |f(x)\rangle. \]
Inverting this circuit effectively ``erases" the index register, which implies that the quantum query complexity of \textsc{Index Erasure} is $\OO(\sqrt{n})$.

The first non-trivial lower bounds on the quantum query complexity of \textsc{Index Erasure} were obtained via the \textsc{Set Equality} problem, which asks to decide whether two injective functions $f,f'$ given via black-box oracles $O_f,O_{f'}$ have the same image or have disjoint images, given a promise that either is the case. \textsc{Set Equality} can be easily reduced to non-coherent (and, thus, coherent too) \textsc{Index Erasure} via the swap test, increasing the number of oracle queries by at most a constant factor. Thus, when Midrij\=anis presented an $\Omega((n/\log n)^{1/5})$ lower bound on the quantum query complexity of \textsc{Set Equality}~\cite{Mid04}, the same lower bound automatically applied to \textsc{Index Erasure}. Ambainis, Magnin, Roetteler, and Roland devised the \emph{hybrid adversary method} \cite{AMRR11}, which they used to prove a tight $\Omega(\sqrt{n})$ lower bound for \textsc{Index Erasure} in the \emph{coherent} regime, and left the non-coherent case as an open question. Later, the lower bound for \textsc{Set Equality} was improved to a tight $\Omega(n^{1/3})$ \cite{Zha15,BR18}, which in turn led to an improved query lower bound for the non-coherent \textsc{Index Erasure}.

The focus of this work is to prove a tight lower bound on the quantum query complexity of \textsc{Index Erasure} in the \emph{non-coherent} case. To show this, we use the so-called \emph{general adversary method}~\cite{LMRSS11} which we review in Section~\ref{genAdv}.

\section{General Adversary Method}\label{genAdv}

The general adversary method places optimal lower bounds on the quantum query complexity of any state conversion problem \cite{LMRSS11}. State conversion problems generalize state generation problems, yet in this paper it will suffice to introduce the adversary bound only for the latter.

The general adversary bound is stated via the $\gamma_2$ and filtered $\gamma_2$ norms,
which are defined as follows. Let $M$ be any matrix and let $\Delta=\{\Delta_x\colon x\in[n]\}$ be a family of matrices of the same dimensions as $M$. Define 
\begin{alignat*}{2}
&\gamma_2(M)&\,:=\,&\max_{\Gamma'}\{\|M\circ \Gamma'\|\colon\|\Gamma'\|\le 1\},\\
&\gamma_2(M|\Delta)&\,:=\,&\max_\Gamma\big\{\|M\circ \Gamma\|\colon\max_{x\in[n]}\|\Delta_x\circ \Gamma\|\le 1\big\},
\end{alignat*}
where $\circ$ denotes the Schur (i.e., entrywise) product of two matrices and, thus, $\Gamma$ and $\Gamma'$ are required to have the same dimensions as $M$.
One can show that $\gamma_2(\cdot)$ is a norm over the set of all matrices and $\gamma_2(\cdot|\Delta)$ is a norm over the set of matrices $M$ that has $M_{f,f'}=0$ whenever $(\Delta_x)_{f,f'}=0$ for all $x\in [n]$ (see \cite{LMRSS11} for details). 
The two norms are called the \emph{$\gamma_2$ norm} and the \emph{filtered $\gamma_2$ norm}, respectively.

The general adversary bound employs various real symmetric matrices whose rows and columns are labeled by black-box functions $f\in\Dom$ in the same order. The family of \emph{difference matrices} $\Delta$ is defined as follows. For each $x\in[n]$, the  $\Delta_x$ is a binary matrix such that $(\Delta_x)_{f,f'}:=1$ if and only if $f(x)\ne f'(x)$.
A \emph{state matrix} is any positive-semidefinite matrix $T$ such that $T\circ I=I$. In other words, it is a Gram matrix corresponding to some family of unit vectors. 
Note that $\gamma_2(\cdot|\Delta)$ is a norm on the set of matrices whose diagonals are all-zeros, and a difference of any two state matrices belongs to this set.

Let $\junkSet$ be the set of all state matrices.  {(In Section~\ref{sec:AdvToKrein}, we will narrow the definition of $\junkSet$ to contain only state matrices possessing certain symmetries.)}
 Note that $\junkSet$ is a compact set and it is closed under the Schur product.
Two particular state matrices of our interest are the all-ones matrix $J$, which corresponds to the family $\{\ket{\mathbf{0}}\colon f\in\Dom\}$, and the \emph{target matrix} $T^\odot$ defined as $(T^\odot)_{f,f'}:=\<\psi_f|\psi_{f'}\>$.

Theorem~\ref{thm:adv} is a special case of \cite[Theorem~4.9]{LMRSS11}.
\begin{thm}\label{thm:adv}
The $\epsilon$-error quantum query complexity of a non-coherent state generation problem with the target matrix $T^\odot$ and the family of difference matrices $\Delta$ is both
\[
\Omega\big(\Adv_{2\sqrt{2\epsilon}}\big)
\quad\text{and}\quad
\OO\big(\Adv_{\epsilon^4/16}\,\epsilon^{-2}\log\epsilon^{-1}\big),
\]
where
\begin{equation}
\label{eq:deltaAdv}
\Adv_\delta:=\min_{R,T\in\junkSet}\{\gamma_2(J-R|\Delta) \colon \gamma_2(R-T^\odot\circ T)\le \delta\}.
\end{equation}
In the case of coherent state generation, one imposes $T=J$ in the expression for $\Adv_\delta$.
\end{thm}

In the expression for $\Adv_\delta$, the state matrix $T$ essentially corresponds to the ancillary states that are prepared in addition to the target states. Thus, assuming there were no error, $T^\odot\circ T $ would be the Gram matrix corresponding to the final states of the whole system. However, since one allows some error---determined by the parameter $\delta$---it suffices that the state matrix $R$ corresponding \emph{exactly} to the final states of the algorithm is close to $T^\odot\circ T $.

When applying the adversary bound, it is convenient to actually apply it to the zero-error case therefore eliminating the matrix $R$ from the consideration. In particular, this leads to the following corollary of Theorem~\ref{thm:adv}. 

A symmetric matrix $\Gamma$ that satisfies $\|\Delta_x\circ \Gamma\|\le 1$ for all $x$ is called an \emph{adversary matrix}.
Let $\Pi_\Gamma$ denote the orthogonal projector on the image of $\Gamma$.

\begin{cor}
\label{cor:adv}
Let $\Gamma$ be an adversary matrix for a non-coherent state generation problem with the target matrix $T^\odot$ and the family of difference matrices $\Delta$, let $\omega$ be a principal eigenvector of $\Gamma$ of norm $1$, and let 
\[
\eta':=\max_{T\in\junkSet} \,\omega^\top(T^\odot \circ T \circ \Gamma/\|\Gamma\|)\,\omega.
\]
The $\epsilon$-error quantum query complexity of the problem is
\[
\Omega\big((1-\eta'-2\sqrt{2\epsilon})\,\|\Gamma\|\,\big).
\]
If $\omega$ is a uniform superposition over $\Dom$, then $\eta'\le\eta$ for
\[
\eta := \max_{T\in\junkSet} \Tr{\Pi_\Gamma(T^\odot \circ T)/|\Dom|}.
\]
\end{cor}

\begin{proof}
For the first part of the corollary, suppose $R,T\in\junkSet$ satisfy $\gamma_2(R-T^\odot \circ T)\le 2\sqrt{2\epsilon}$ and are thus a feasible solution to the minimization in $\Adv_{2\sqrt{2\epsilon}}$. We have
\begin{align*}
\gamma_2(J-R|\Delta) \ge & \|(J-R)\circ\Gamma\| \\
\ge & \|(J-T^\odot\circ T)\circ\Gamma\| - \|\Gamma\|\big\|(R-T^\odot\circ T)\circ\Gamma/\|\Gamma\|\big\| \\
\ge & \,\omega^\top\Gamma\,\omega - \omega^\top(T^\odot \circ T \circ \Gamma)\, \omega - 2\sqrt{2\epsilon}\|\Gamma\| \\
\ge & (1-\eta'-2\sqrt{2\epsilon})\|\Gamma\|.
\end{align*}
For the second part, note that, if $\omega$ is a uniform superposition over $\Dom$, then, for any two symmetric $|\Dom|\times|\Dom|$ matrices $M,M'$, we have $\omega^\top(M\circ M')\,\omega=\Tr{MM'}/|\Dom|$.
The inequality $\eta'\le \eta$ results from both $T^\odot\circ T$ and $\Pi_\Gamma-\Gamma/\|\Gamma\|$ being positive-semidefinite.
\end{proof}

\subsection{Automorphism Principle for State Generation} \label{sec:invar}

The \emph{automorphism principle} of \cite{HLS07} addresses the adversary bound for function evaluation problems and states that, without loss of generality, the optimal adversary matrix can be required to respect symmetries of the problem. The main result of this section is Theorem~\ref{thm:auto},  {which is a generalization of the automorphism principle} to state generation problems. It is not difficult to see that the proof of Theorem~\ref{thm:auto} can be generalized further to state conversion problems \emph{mutatis mutandis}; however, since the current application to \textsc{Index-Erasure} is a state generation problem, we have elected not to prove Theorem~\ref{thm:auto} in this generality.

The wreath product $S_m\wr S_n$ of groups $S_m$ and $S_n$ is the group whose elements are
$(\pi,\bm{\sigma})\in S_n\times S_m^n$ and whose group operation is
\[
\big(\pi',(\sigma_1',\cdots,\sigma_n')\big)
\big(\pi,(\sigma_1,\cdots,\sigma_n)\big)
=
\big(\pi'\pi,(\sigma_1'\sigma_{(\pi')^{-1}(1)},\cdots,\sigma_n'\sigma_{(\pi')^{-1}(n)})\big)
\]
(see~\cite[Ch. 4]{JamesKerber}). Similarly to (\ref{eq:SnSmAction}) below, the action of $S_m\wr S_n$ on
$f\colon[n]\rightarrow[m]$ is given by
\begin{equation}\label{eq:WreathAction}
\big((\pi,\bm{\sigma})f\big)(x) = \sigma_x(f(\pi^{-1}(x))\qquad\text{for all }x\in[n].
\end{equation}
The action of a subgroup $G\le S_m\wr S_n$ on the set of black-box functions $\Dom$ is \emph{closed} if $g(f)\in\Dom$ for all $f\in\Dom$ and $g\in G$.

Suppose $M$ is a symmetric $|\Dom|\times|\Dom|$ matrix whose rows and columns are labeled by $f\in\Dom$ in the same order and suppose the 
action of a subgroup $G\le S_m\wr S_n$ on $\Dom$ is closed. We say that $M$ is \emph{$G$-invariant} if $M_{g(f),g(f')}=M_{f,f'}$ for all $f,f'\in\Dom$  and $g\in G$. Similarly, a vector $\omega\in\mathbb{C}[\Dom]$ is $G$-invariant if $\omega_{g(f)}=\omega_f$ for all $f\in\Dom$  and $g\in G$.
A subgroup $G$ is an \emph{automorphism group} for a state generation problem with a target matrix $T^\odot$ if  $G$'s action on $\Dom$ is closed and $T^\odot$ is $G$-invariant.%
\footnote{The $G$-invariance of $T^\odot$ is equivalent to the existence of a unitary representation $U_g$ of $G$ acting on the target space such that $U_g|\psi_f\>=|\psi_{g(f)}\>$ for all $f\in\Dom$ and $g\in G$.}

Note that the free product of two automorphism groups is an automorphism group, so one can consider the maximum automorphism group of a problem. For example, the maximum automorphism group of \textsc{Parity} is the whole wreath product $S_2\wr S_n$ while the maximum automorphism groups of \textsc{Or} and \textsc{Index Erasure} are, respectively, 
\begin{align*}
& \{(\pi,(\varepsilon,\cdots,\varepsilon))\colon \pi\in S_n \} \cong S_n,\\
& \{(\pi,(\sigma,\cdots,\sigma))\colon \pi\in S_n \text{ and }\sigma\in S_m\} \cong S_n\times S_m,
\end{align*}
where $\varepsilon$ is the identity permutation in $S_2$.
 {Note that for \textsc{Parity}, the target states $|\psi_{f}\>$ and $|\psi_{g(f)}\>$ may differ for $g$ in the maximum automorphism group, and the same is true for  \textsc{Index Erasure}.}

\begin{thm}
\label{thm:auto}
Let $G$ be an automorphism group for a non-coherent state generation problem. The value of $\Adv_\delta$ remains the same if one restricts the minimization in the expression defining $\Adv_\delta$ and the maximization in the expressions defining the $\gamma_2$ and filtered $\gamma_2$ norms to $R,T,\Gamma,\Gamma'$ that are all $G$-invariant. 
\end{thm}

The proof of Theorem~\ref{thm:auto} splits into two parts according to the two types of symmetrizations of the matrices $R,T,\Gamma,\Gamma'$, which depend on whether they are arguments in the aforementioned minimization or maximization.
\begin{proof}[Proof of Theorem~\ref{thm:auto}]
	Let $M$ be a generic symmetric matrix whose rows and columns are labeled by black-box functions $f\in\Dom$ in the same order.
	Let $g(M)$ be obtained by permuting the rows and the columns of $M$ according to the action of $g\in G$ of $\Dom$ (see (\ref{eq:WreathAction})). Namely, entrywise we define $g(M)$ as
	\[
	(g(M))_{f,f'} := M_{g^{-1}(f),g^{-1}(f')}.
	\]
	Similarly, for a vector $\omega\in\mathbb{C}[\Dom]$, define $g(\omega)$ entrywise as
	\(
	(g(\omega))_{f'} := \omega_{g^{-1}(f)}.
	\)
	For the sake of conciseness, we also occasionally write $M^g$ and $\omega^g$ instead of $g(M)$ and $g(\omega)$, respectively.
	Note that $M$ is $G$-invariant if $M^g=M$ for all $g\in G$, and $T^\odot,I,J$ are $G$-invariant.
	Also note that $(M\circ M')^g=M^g\circ M'^g$.
	
	Let $\Delta=\{\Delta_1,\cdots,\Delta_n\}$ be the family of difference matrices.
	This family is closed under the action of $G$ is the following sense.
	\begin{clm}\label{claim:DeltaClosed}
		We have $(\pi,\bm{\sigma})(\Delta_x)=\Delta_{\pi(x)}$ for all $(\pi,\bm{\sigma})\in G$.
	\end{clm}
	\begin{proof}
		Fix $(\pi,\bm{\sigma})\in G$ and let $g:=(\pi,\bm{\sigma})^{-1}$. Note that $g=(\pi^{-1},\bm{\sigma}')$ for some $\bm{\sigma}'\in S_m^n$.
		From (\ref{eq:WreathAction}), we have 
		\(
		(g(f))(x) = (g(f'))(x)
		\)
		if and only if
		\(
		f(\pi(x)) = f'(\pi(x)).
		\)
		As a result, we have
		\[
		((\pi,\bm{\sigma})(\Delta_x))_{f,f'} = (\Delta_x)_{g(f),g(f')} = 1
		\]
		if and only if 
		$f(\pi(x))=f'(\pi(x))$.
	\end{proof}
	\noindent Note that $M^g$ equals $M$ with its rows and columns permuted. Permuting rows and columns does not affect the $\gamma_2$ norm, so we have $\gamma_2(M^g)=\gamma_2(M)$ for all $g\in G$. And, if the diagonal of $M$ is all-zeros, then Claim~\ref{claim:DeltaClosed} also implies that $\gamma_2(M^g|\Delta) = \gamma_2(M|\Delta)$ for all $g\in G$.

	\begin{clm}
		Restricting $R,T\in\junkSet$ to be $G$-invariant does not change the optimal value of the minimization problem defining $\Adv_\delta$.
	\end{clm}

	\begin{proof}
		Let $R,T$ be an optimal solution of the minimization in (\ref{eq:deltaAdv}). We define their respective $G$-symmetrizations as
		\begin{equation*}
		\overline{R}:=\frac{1}{|G|}\sum_{g\in G} g(R)
		\qquad\text{and}\qquad
		\overline{T}:=\frac{1}{|G|}\sum_{g\in G} g(T),
		\end{equation*}
		which are both clearly in $\junkSet$. 
		Since $g(T^\odot)=T^\odot$ for all $g\in G$, the triangle inequality yields
		\begin{multline*}
		\gamma_2(\overline{R}-T^\odot\circ\overline{T})
		= \gamma_2\bigg(\frac{1}{|G|}\sum_{g\in G}g(R-T^\odot\circ T)\bigg) 
		\le \frac{1}{|G|}\sum_{g\in G} \gamma_2\big(g(R-T^\odot\circ T)\big) \\
		= \frac{1}{|G|}\sum_{g\in G} \gamma_2(R-T^\odot\circ T)
		= \gamma_2(R-T^\odot\circ T) \le \delta.
		\end{multline*}
		Hence we have show that the pair $\overline{R},\overline{T}$ is a feasible solution to the minimization in (\ref{eq:deltaAdv}), and it remains to show that it is also optimal.
		And, again by the triangle inequality,
		\begin{multline*}
		\gamma_2(J-\overline{R}|\Delta)
		= \gamma_2\bigg(\frac{1}{|G|}\sum_{g\in G}g(J-R)\bigg|\Delta\bigg)
		\le \frac{1}{|G|}\sum_{g\in G} \gamma_2\big(g(J-R)|\Delta\big)
		\\
		= \frac{1}{|G|}\sum_{g\in G} \gamma_2(J-R|\Delta)
		= \gamma_2(J-R|\Delta)=\Adv_\delta.
		\end{multline*}
	\end{proof}
	
	Now, fix $G$-invariant $R,T\in\junkSet$ and let $M:=J-R$, which is also $G$-invariant. 
	Let us now show that the maximization in 
	\[
	\gamma_2(M|\Delta) = \max_\Gamma \big\{\|M\circ \Gamma\|\colon \forall x\; \|\Delta_x\circ \Gamma\|\le 1\big\}
	\]
	can be restricted to $G$-invariant $\Gamma$.
	
	The proof now proceeds along the lines of the automorphism principle in~\cite{LMRSS11}.
	Fix an optimal solution $\Gamma$, and without loss of generality assume that the largest eigenvalue of $M\circ \Gamma$ is positive and let it correspond to an eigenvector $\omega\in\mathbb{C}[\Dom]$ of norm $1$. Namely, 
	\[
	\|M\circ\Gamma\| = \omega^\top\!(M\circ \Gamma)\omega.
	\]
	Define the $G$-symmetrization $\overline\omega$ of $\omega$ entrywise as
	\[
	\overline\omega_f := \sqrt{\frac{1}{|G|}\sum_{g\in G}|(\omega^g)_f|^2},
	\]
	and note that $\overline\omega$ also has norm $1$.
	Without loss of generality, all the entries of $\overline\omega$ are strictly positive (the rows and columns corresponding to $f$ such that $\overline\omega_f=0$ can be removed from the consideration), and thus we can entrywise define a vector $\mu$ as
	$\mu_f:=1/{\overline\omega_f}$.
	Let us define
	\[
	\overline{\Gamma}:= \mu\mu^\top\circ \frac{1}{|G|}\sum_{g\in G} \Gamma^g\circ \omega^g\omega^{g\top},
	\]
	which is clearly $G$-invariant.
	
	Let us start by showing that $\|\overline\Gamma\circ\Delta_x\|\le 1$ for all $x$. Note that
	$\|\Gamma^g\circ \Delta_x\|\le1$ for all $x$ and all $g\in G$ due to Claim~\ref{claim:DeltaClosed},
	$\|\Gamma\circ\Delta_x\|\le 1$ if and only if $I\pm \Gamma\circ \Delta_x $ is positive-semidefinite,
	and
	\[
	I\circ \mu\mu^\top\circ \frac{1}{|G|}\sum_{g\in G}  \omega^g\omega^{g\top}=I.
	\]
	We thus have that
	\begin{align*}
	I\pm \overline\Gamma\circ\Delta_x
	= &\,
	\mu\mu^\top\circ \frac{1}{|G|}\Big(
	\sum_{g\in G} \omega^g\omega^{g\top} \circ (I\pm\Gamma^g\circ \Delta_x)
	\Big)
	\end{align*}
	is positive-semidefinite as the sum and the entrywise product of positive-semidefinite matrices are positive-semidefinite.
	Thus, indeed, $\|\overline\Gamma\circ\Delta_x\|\le 1$ for all $x$.
	
	Now let us use the fact that $\omega$ is a principal eigenvector of $M\circ\Gamma$, and, therefore, $\omega^g$ is a principal eigenvector of $M\circ\Gamma^g$ for all $g\in G$ (recall that $M$ is $G$-invariant).
	We have
	\begin{align*}
	\|M\circ\overline{\Gamma}\| \ge \overline{\omega}\,(M\circ\bar\Gamma)\,\overline{\omega}^\top
	= &\, \sum_{f,f'\in\Dom} \Big(\frac{1}{|G|}\sum_{g\in G}(M\circ\Gamma^g\circ \omega^g\omega^{g\top})\Big)_{f,f'}
	\\ = &\,
	\frac{1}{|G|}\sum_{g\in G}\omega^{g\top}(M\circ\Gamma^g) \omega^g
	= \|M\circ\Gamma\|.
	\end{align*}
	Thus $\overline{\Gamma}$ is also an optimal solution of the maximization above. Also note that $\overline\omega$ is the principal eigenvector of $M\circ\overline\Gamma$.
	
	A similar argument shows that, for $G$-invariant $M':=R-T^\odot\circ T$, one can restrict the maximization in
	\[ 
	\gamma_2(M') = \max_{\Gamma'} \big\{\|M'\circ \Gamma'\|\colon \|\Gamma'\|\le 1\big\}
	\]
	to $G$-invariant $\Gamma'$. This completes the proof of Theorem~\ref{thm:auto}.
\end{proof}

Note that the ability to restrict $T$ and $\Gamma$ to be $G$-invariant carries over from Theorem~\ref{thm:adv} to Corollary~\ref{cor:adv}. The ability to restrict $T$ will be paramount in our proof (see Section~\ref{sec:AdvToKrein}). On the other hand, the ability to restrict $\Gamma$ is optional. Namely, Corollary~\ref{cor:adv} provides an adversary bound regardless of what restrictions one imposes on $\Gamma$, yet for too strict restrictions this bound would not be optimal.

As observed in \cite{AMRR11}, the set of $|\Dom| \times |\Dom|$ matrices indexed by $\Dom$ that are $(S_n \times S_m)$-invariant under the aforementioned action (\ref{eq:WreathAction}) afford a commutative matrix algebra. In particular, it is the Bose--Mesner algebra of a symmetric association scheme defined over injections, which we formally define in Section~\ref{sec:assoc}. Before we define this association scheme, some results from the representation theory of the symmetric group are needed, which we overview in the next section.

\section{Representation Theory Preliminaries}\label{sec:prelims}

We refer the reader to~\cite{Diaconis88} for an introduction to group representation theory, \cite{Sagan} for more details on the representation theory of the symmetric group, and~\cite{CST} for a more involved discussion on finite Gelfand pairs and their spherical functions.

Let $\text{Sym}(X)$ denote the \emph{symmetric group} on the symbol set $X$. If $X = [m] := \{1,2,\cdots,m\}$, then we define $S_m := \text{Sym}(X)$. It is well-known that the conjugacy classes of $S_m$ and irreducible representations (irreps) of $S_m$ are given by the cycle-types of permutations of $S_m$, which in turn are in one-to-one correspondence with \emph{integer partitions} $\lambda \vdash m$, i.e., $\lambda := (\lambda_1, \lambda_2, \cdots , \lambda_k) \vdash m \text{ such that } \lambda_1 \geq \lambda_2 \geq \cdots \geq \lambda_k \geq 0\text{ and } \sum_{i=1}^k \lambda_i = m.$  
We may visualize $\lambda$ as a \emph{Young diagram}, a left-justified table of cells that contains $\lambda_i$ cells in the $i$th row. When referencing a Young diagram, we alias $\lambda$ as the \emph{shape}.  A \emph{standard Young tableau} of shape $\lambda \vdash n$ is a Young diagram with unique entries from $[n]$ that are strictly increasing along rows and strictly increasing along columns.
For example, the left Young diagram below has shape $(5,3,2,1) \vdash 11$ and the tableau on the right is a standard Young tableau of the same shape
\[
	\ytableausetup{mathmode,boxsize=1.0em}
	\begin{ytableau}
		\empty  & \empty &  \empty & \empty & \empty \\
		\empty & \empty & \empty \\ 
		\empty & \empty  \\ 
		\empty \\  
	\end{ytableau}
\quad \quad \quad \quad \quad \quad \quad \quad \quad 
	\ytableausetup{mathmode,boxsize=1.2em}
	\begin{ytableau}
		1 & 2 & 5 & 8 & 9 \\
		3 & 6 & 7\\ 
		4 & 10  \\ 
		11 \\  
	\end{ytableau}~.
\]
Let $\cV_\lambda$ denote the $S_m$-irrep corresponding to $\lambda \vdash m$. Let $d_\lambda$ be the number of standard Young tableau of shape $\lambda$. It is well-known that $d_\lambda$ is also the dimension of the $S_m$-irrep corresponding to $\lambda \vdash m$. The number of standard Young tableau can be counted elegantly via the \emph{hook rule} (see~\cite{Sagan} for a proof).
\begin{thm}[Hook rule]
	Let $\lambda \vdash m$, and for any cell $c \in \lambda$ of the Young diagram of $\lambda$ define the \emph{hook-length} $\sh_\lambda(c)$ to be the total number of cells below $c$ in the same column and to the right of $c$ in the same row, plus 1.  Then we have 
	$d_\lambda = {m!}/{\hooks{\lambda}}  \text{ where }\hooks{\lambda}:=\prod_{c \in \lambda} \sh_\lambda(c).$
\end{thm}
Another well-known result is the \emph{branching rule}, which describes how an $S_m$-irrep decomposes into $(S_{m-1})$-irreps (see~\cite{Sagan} for a proof). We say that a cell of a Young diagram is an \emph{inner corner} if it has no cells to its right and no cells below it.
\begin{thm}[The Branching Rule]
	If $\cV_\lambda$ is  {an} $S_m$-irrep, then $\cV_\lambda \cong \bigoplus_{\lambda^-} \cV_{\lambda^-}$
	where $\lambda^-$ ranges over all shapes obtainable by removing an inner corner from $\lambda$
 { and $\cV_{\lambda^-}$ is an $(S_{m-1})$-irrep corresponding to $\lambda^-$.}
\end{thm}

\noindent The hook rule and the branching rule can be used to prove the following results. 
 {For any $\lambda\vdash n$, recall that $\bar\lambda\vdash m$ is obtained from $\lambda$ by adding $m-n$ cells to the first row of $\lambda$.}

\begin{prop}\label{prop:dims}\emph{\cite{Diaconis88}} Let $\lambda  \vdash n$ and $\ell = n - \lambda_1$. Then we have $d_{\lambda}d_{\bar{\lambda}} \leq   \binom{n}{\ell}\binom{m}{\ell} \ell! \leq m^\ell n^\ell$.
\end{prop}	

\begin{thm}
\label{thm:dimratio}
	Let $\theta \vdash k$ and $\theta^+ \vdash (k+1)$ be any shape obtained by adding an inner corner to $\theta$. For all $m \geq 2(k + 1)$, we have 
	\[\frac{d_{(m - k - 1, \theta^+)}}{d_{(m - k, \theta)}} \geq  \frac{m}{k}\cdot\left(1-\frac{2k+1}{m}\right).\]
\end{thm}
\begin{proof}
Recall that, by the hook rule, $\hooks{\theta}d_\theta=|\theta|!$. First, \cite[Claim 6.3]{Ros14} states that
\[
\frac{d_{(m-|\theta^+|,\theta)}}{d_{(m-|\theta|,\theta)}}\geq 1-\frac{2k}{m}.
\]
We reprove this claim here for completeness.
Note that when we add a cell to the end of the top row of $(m-|\theta^+|,\theta)$ to obtain $(m - |\theta|,\theta)$, this increases the hook-lengths of the cells in the top row by 1, and the rest of the hook-lengths are unchanged.  If we just consider the ``overhang" and ignore everything else in the first row, then the product of the hook-lengths with respect to $(m - |\theta|,\theta)$ is $(m-2k)!$ whereas it is $(m-2k-1)!$ with respect to $(m-|\theta^+|,\theta)$.  This gives us
$$\frac{d_{(m-|\theta^+|,\theta)}}{d_{(m-|\theta|,\theta)}}\geq  \frac{m-2k}{m} = 1 - \frac{2k}{m},$$
which proves the claim.

Since $(m-|\theta^+|,\theta)$ is a partition of $(m-1)$, it corresponds to an $(S_{m-1})$-irrep.
When we added one cell to $(m-|\theta^+|,\theta)$ to obtain $(m-|\theta^+|,\theta^+)$, only one hook-length in the first row increased, and before the increment it was at least $m-2k-1$. Thus, in the following derivation, all the other hook-lengths of the first rows of $(m-|\theta^+|,\theta)$ and $(m-|\theta^+|,\theta^+)$ have cancelled out.
\begin{align*}
\frac{d_{(m-|\theta^+|,\theta^+)}}{d_{(m-|\theta^+|,\theta)}} & \geq
\frac{m!}{(m-1)!}\cdot\frac{m-2k-1}{m-2k}\cdot\frac{\hooks{\theta}}{\hooks{\theta^+}} \\& =
m\left(1-\frac{1}{m-2k}\right)\frac{k!d_{\theta^+}}{(k+1)!d_\theta} \geq \frac{m}{k+1}\left(1-\frac{1}{m-2k}\right),
\end{align*}
 {where the middle equality uses the hook rule once more}, and the last inequality follows from the branching rule (namely, that $d_{\theta^+} \geq d_\theta$). Combining these two inequalities gives the result.
\end{proof}

\noindent
Note that the above fraction is greater than $1$ when $m>3k+1$.

\subsection{The Representation Theory of Injections}

Henceforth, let $S_{n,m}$ denote the collection of injective maps $f : [n] \rightarrow [m]$, equivalently, $n$-tuples $f := (f(1),f(2),\cdots,f(n))$
with no repeated elements such that $f(x) \in [m]$ for all $x \in [n]$. The latter representation is a bit more succinct, so we shall prefer it for describing particular injections.
When $m=n$ we recover the symmetric group $S_n$ on $n$ symbols. To understand the representation theory of $S_{n,m}$ we must first broaden our Young tableau vocabulary. 

For any $\lambda \vdash m$, let $\len{\lambda}$ denote the \emph{length} of $\lambda$, that is, the number of parts in the partition. For any integer partition $\mu$, we let $|\mu|$ denote the size of $\mu$, i.e., number of cells in its Young diagram.  We say that a shape $\lambda$ \emph{covers} a shape $\mu$ if $\mu_i \leq \lambda_i$ for each $i$.  If $\lambda$ and $\mu$ are two shapes such that $\lambda$ covers $\mu$, then we obtain the \emph{skew shape} \emph{$\lambda/\mu$} by removing the cells corresponding to $\mu$ from $\lambda$.  For instance, the shape $(5,3,2,1)$ covers $(2,2,1)$, so we may consider the skew shape $(5,3,2,1)/(2,2,1)$:
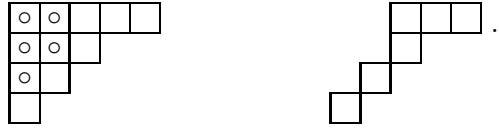
\begin{center}
	\quad\quad\quad\quad\quad\quad\quad\quad\quad\quad\quad\quad
	\ytableausetup{mathmode,boxsize=1.0em}
	\begin{ytableau}
		\circ  & \circ &  \empty & \empty & \empty \\
		\circ & \circ & \empty\\ 
		\circ & \empty  \\ 
		\empty \\  
	\end{ytableau}
	\hfill
	\ytableausetup{mathmode,boxsize=1.0em}
	\begin{ytableau}
		\none & \none &  \empty & \empty & \empty \\
		\none & \none & \empty\\ 
		\none & \empty  \\ 
		\empty \\  
	\end{ytableau}~.
	\quad\quad\quad\quad\quad\quad\quad\quad\quad\quad\quad\quad
\end{center}
A skew shape is a \emph{horizontal strip} if each column has no more than one cell.  For example, the skew shape $(5,3,2,1)/(3,3,1)$ is a horizontal strip, but the skew shape above is not.
\\

\noindent We now observe that Theorem~\ref{thm:dimratio} has the following corollary. 

\begin{cor}\label{cor:lambdabarPre}
	Let $\lambda\vdash n$ be a shape such that $\lambda_1\ge n-\sqrt{n}$ and 
	let $\lambda'\vdash m$ be any shape that covers $\lambda$ such that $\lambda'/\lambda$ is a non-empty horizontal strip.
	Then $d_{\lambda'}/d_{\bar\lambda}  \in \Omega(m/\sqrt{n})$.
\end{cor}

Henceforth, we let $S_n \times S_m$ act on $S_{n,m}$ as follows:
\begin{equation}
\label{eq:SnSmAction}
 (\tau,\sigma) \cdot (f_1,
 \cdots,f_n) = (\sigma(f_{\tau^{-1}(1)}), 
  \cdots, \sigma(f_{\tau^{-1}(n)}))\text{ for all } (\tau,\sigma) \in S_n \times S_m.
 \end{equation}
The stabilizer of the \emph{identity injection} $f_{\text{id}} := (1,2,\cdots,n) \in S_{n,m}$ in $S_n \times S_m$ is isomorphic to the group 
$$\text{diag}(S_n \times S_n) \times S_{m-n} = \{(\tau,\tau,\pi) : \tau \in \text{Sym}([n]), \pi \in \text{Sym}(\{n+1,\cdots,m\}).$$ 
One can show (see~\cite{CST}) that the permutation representation of $(S_n \times S_m)$ acting on $S_{n,m} \cong (S_n \times S_m)/(\text{diag}(S_n \times S_n) \times S_{m-n})$ is \emph{multiplicity-free}, that is, its decomposition has at most one copy of any $(S_n \times S_m)$-irrep, as shown in Theorem~\ref{thm:decomp}.   
\begin{thm}\label{thm:decomp}
	\emph{\cite{CST}} The complex-valued functions over $S_{n,m}$ denoted as $\mathbb{C}[S_{n,m}]$ admits the following decomposition into $(S_n \times S_m)$-irreps:
	\[  \mathbb{C}[S_{n,m}] \cong \bigoplus_{\mu,\lambda} \cV_\mu \otimes \cV_\lambda \]
	where $\mu,\lambda$ ranges over all pairs $\mu \vdash n,\lambda \vdash m$ such that $\lambda / \mu$ is a horizontal strip.  
\end{thm}
Let $\text{Irr}(S_{n,m})$ denote the set of $(S_n \times S_m)$-irreps that appear in Theorem~\ref{thm:decomp}. Every multiplicity-free permutation representation gives rise to a commutative association scheme (see~\cite{BannaiI84}), so a consequence of Theorem~\ref{thm:decomp} is the existence of a symmetric association scheme $\mathcal{A}_{n,m}$ over $S_{n,m}$ that we call \emph{the injection association scheme}~\cite{Munemasa01}. In Section~\ref{sec:assoc}, we discuss this association scheme in more detail.

For any finite groups $K \leq G$, we say that $(G,K)$ is a \emph{finite Gelfand pair} if its double coset algebra $\mathbb{C}[K \backslash G / K]$ is commutative, or equivalently, if the permutation representation of $G$ acting on $G/K$ is multiplicity-free. For any Gelfand pair $(G,K)$, let $\omega^i$ be the \emph{spherical function} corresponding to irrep indexed by $i$, i.e., the projection of the irreducible character indexed by $i$ onto the space of (left) $K$-invariant functions of $\mathbb{C}[G/K]$~(see \cite{CST}). The spherical functions are constant on double cosets $K\backslash G / K$, so we may define $\omega^i_j$ to be the evaluation of $\omega^i$ on the double coset indexed by $j$. For more details on the connection between Gelfand pairs and association schemes, see~\cite{BannaiI84}.\\

\noindent We recall two well-known and basic facts about the spherical functions of finite Gelfand pairs.
\begin{prop}\emph{\cite{CST}}\label{prop:trivial} For any spherical function $\omega^i$ and double coset $j$, we have $|\omega^i_j| \leq 1$.
\end{prop}
\noindent A finite Gelfand pair $(G,K)$ is \emph{symmetric} if $g^{-1}\in KgK$ for all $g \in G$. Define $\delta_{i,j}$ so that $\delta_{i,j} = 1$ if $i=j$; otherwise, $\delta_{i,j} = 0$.
\begin{prop}\emph{\cite{CST}}\label{prop:orth} Let $X = G/K$ such that $(G,K)$ is a finite symmetric Gelfand pair. Let $\omega^i$ denote the $i$-th spherical function corresponding to irrep $i$ of dimension $d_i$. Then 
	\[
	\sum_{x \in X} \omega^i(x) \overline{\omega^j(x)}  = \sum_{x \in X} \omega^i(x) \omega^j(x) = \delta_{i,j} \frac{|X|}{d_i}.
	\]
\end{prop}
\noindent It is known that $(S_n \times S_m,\text{diag}(S_n \times S_n) \times S_{m-n})$ is a finite symmetric Gelfand pair (see \cite{CST}), and we will use these basic results in our proof of the main result. 

Finally, as stated in the Section~\ref{ssec:outline}, the minimal 
irreps of $\text{Irr}(S_{n,m})$ will be of particular importance in our proof of the main result, which we formally define below.

\begin{definition}[Minimal Irreps]
	For any $\lambda \vdash n$, the \emph{minimal irrep} 
	with respect to $\lambda$ is $\lambda \otimes {\bar \lambda}$. 
\end{definition}
\noindent See (\ref{eq:minmax}) in Section~\ref{ssec:outline} for a graphical example. Note that if $\lambda_1\ge n-\sqrt{n}$, then Theorem~\ref{thm:dimratio} implies that the minimal irreps
indeed have the least 
dimension over all irreps of the form $\lambda \otimes \mu \in \text{Irr}(S_{n,m})$ for sufficiently large $m$.

\section{The Injection Association Scheme}\label{sec:assoc}

The theory of association schemes will be a convenient language for describing the algebraic and combinatorial components of our work. We refer the reader to Bannai and Ito's reference~\cite{BannaiI84} and Chris Godsil's notes~\cite{GodsilAssoc} for a more thorough treatment.

\begin{definition}[Association Schemes]
A \emph{symmetric association scheme} is a collection of $d+1$ binary $|X| \times |X|$ matrices $\mathcal{A} = \{A_0,A_1,\cdots,A_d\}$ over a set $X$ that satisfy the following axioms:
\begin{enumerate}
\item $A_i$ is symmetric for all $0 \leq i \leq d$,
\item $A_0 = I$ where $I$ is the identity matrix,
\item $\sum_{i=0}^d A_i = J$ where $J$ is the all-ones matrix, and 
\item $A_iA_j  = A_jA_i  \in \text{Span}\{A_0, A_1,\cdots, A_d\} =: \mathfrak{A}$ for all $0 \leq i,j \leq d$.
\end{enumerate}
The matrices $A_1,A_2,\cdots,A_d$ are called the \emph{associates}, and the algebra $\mathfrak{A}$ is called the \emph{Bose--Mesner} algebra of the association scheme. Moreover, $\mathfrak{A}$ admits a unique dual basis of \emph{primitive idempotents} $E_0,E_1,\cdots,E_d$, i.e., $E_i^2 = E_i$ for all $0 \leq i \leq d$ and $\sum_{i=0}^d E_i = I$.  
\end{definition}

Since the permutation representation of $S_n \times S_m$ acting on $S_{n,m}$ is multiplicity-free (see Theorem~\ref{thm:decomp}), the orbits $A_0, A_1, \cdots, A_d$ (so-called \emph{orbitals}) of the action of $S_n \times S_m$ on ordered pairs $S_{n,m} \times S_{n,m}$ forms a \emph{symmetric association scheme} (see~\cite{BannaiI84} for a proof). 
We abuse the notation, and also use $A_i$ to denote the binary matrix with entries $1$ corresponding to exactly those pairs that are in the orbit $A_i$.
Let $\mathcal{A}_{n,m} := \{I, A_1, \cdots, A_d\}$ denote the \emph{$n,m$-injection association scheme}. 

Although it is well-known that permutation representation of $S_n \times S_m$ acting on $S_{n,m}$ is multiplicity-free (see~\cite{AMRR11,CST,Munemasa01,Greenhalgh} for example), the parameters of its corresponding association scheme $\mathcal{A}_{n,m}$ have not yet been fully worked out. We now give a more in-depth treatment of the injection association scheme.

\subsection{The Associates}
The following is a more combinatorial definition of the associates of $\mathcal{A}_{n,m}$ that gives a combinatorial bijection between the associates of $\mathcal{A}_{n,m}$ and $\text{Irr}(S_{n,m})$, which are the eigenspaces of the association scheme. The bijection is readily observed by thinking of each element of $S_{n,m}$ graphically as a maximum matching of the complete bipartite graph $K_{n,m}$ (see Figure~\ref{fig:graphs}).

Recall that $f_{\text{id}} = (1,2,\cdots,n)$ is the identity injection, which we can view as the maximum matching of $K_{n,m}$ that pairs 1 with 1, 2 with 2, and so on (e.g., the red matching in Figure~\ref{fig:graphs}). 
For any two maximum matchings $f,f'$ of $K_{n,m}$, let $G(f,f')$ be the multigraph whose edge multiset is the multiset union $f \cup f'$.  Clearly $G(f,f') = G(f',f)$ and this graph is composed of disjoint even cycles and disjoint even paths.  Let $c$ denote the number of disjoint cycles and let $2\lambda_i$ denote the length of an even cycle. Let $p$ denote the number of disjoint paths and let $2\rho_i$ denote the length of an even path. If we order the cycles and paths respectively from longest to shortest and divide each of their lengths by two, assuming $m\ge 2n$, we see that the graphs $G(f,f')$ are in bijection (up to graph isomorphism) with pairs $(\lambda | \rho)$ of integer partitions 
$\lambda = (\lambda_1, \lambda_2, \cdots , \lambda_c), \rho = (\rho_1, \rho_2, \cdots , \rho_p) \text{ such that } (\lambda_1, \cdots , \lambda_c, \rho_1, \cdots, \rho_p) \vdash n.$ 
Let $\type{f,f'} := (\lambda | \rho)$ denote this bijection, which we refer to as the \emph{cycle-path type of $f'$ with respect to $f$}. 
Note that $\type{\sigma(f),\sigma(f')}=\type{f,f'}$ for all injections $f,f'$ and all $\sigma\in S_n\times S_m$.
If one of the arguments is the identity matching, then we say $\type{f}:=\type{f_{\text{id}},f}$ is \emph{the cycle-path type of $f$}. Illustrations of the graphs $G_{(\varnothing | n)}$ and $G_{(n-1|1)}$, and $G_{(\varnothing|1^n)}$ are provided in Figure~\ref{fig:graphs} where $n=3$ and $m = 6$. 
\begin{center}
\begin{figure}[h!]
	\centering
\begin{tikzpicture}[style=thick,scale=0.65]

\node[draw,circle,inner sep=0.1cm] (1) at (0,3) [thick] {\textbf{1}};
\node[draw,circle,inner sep=0.1cm] (1') at (2,3) [thick] {\textbf{1}};
\node[draw,circle,inner sep=0.1cm] (2) at (0,2) [ thick] {\textbf{2}};
\node[draw,circle,inner sep=0.1cm] (2') at (2,2) [thick] {\textbf{2}};
\node[draw,circle,inner sep=0.1cm] (3) at (0,1) [thick] {\textbf{3}};
\node[draw,circle,inner sep=0.1cm] (3') at (2,1) [ thick] {\textbf{3}};
\node[draw,circle,inner sep=0.1cm] (4) at (2,0) [ thick] {\textbf{4}};
\node[draw,circle,inner sep=0.1cm] (5) at (2,-1) [ thick] {\textbf{5}};
\node[draw,circle,inner sep=0.1cm] (6) at (2,-2) [ thick] {\textbf{6}};

\draw[red,line width = 3] (1) -- (1');
\draw[red,line width = 3] (2) -- (2');
\draw[red,line width = 3] (3) -- (3');
\draw[blue,dotted,line width = 3] (1) -- (2');
\draw[blue,dotted,line width = 3] (2) -- (3');
\draw[blue,dotted,line width = 3] (3) -- (6);
\end{tikzpicture}
\quad \quad \quad \quad 
\begin{tikzpicture}[style=thick,scale=0.65]

\node[draw,circle,inner sep=0.1cm] (1) at (0,3) [thick] {\textbf{1}};
\node[draw,circle,inner sep=0.1cm] (1') at (2,3) [thick] {\textbf{1}};
\node[draw,circle,inner sep=0.1cm] (2) at (0,2) [ thick] {\textbf{2}};
\node[draw,circle,inner sep=0.1cm] (2') at (2,2) [thick] {\textbf{2}};
\node[draw,circle,inner sep=0.1cm] (3) at (0,1) [thick] {\textbf{3}};
\node[draw,circle,inner sep=0.1cm] (3') at (2,1) [ thick] {\textbf{3}};
\node[draw,circle,inner sep=0.1cm] (4) at (2,0) [ thick] {\textbf{4}};
\node[draw,circle,inner sep=0.1cm] (5) at (2,-1) [ thick] {\textbf{5}};
\node[draw,circle,inner sep=0.1cm] (6) at (2,-2) [ thick] {\textbf{6}};

\draw[red,line width = 3] (1) -- (1');
\draw[red,line width = 3] (2) -- (2');
\draw[red,line width = 3] (3) -- (3');
\draw[blue,dotted,line width = 3] (1) -- (2');
\draw[blue,dotted,line width = 3] (2) -- (1');
\draw[blue,dotted,line width = 3] (3) -- (5);
\end{tikzpicture}
\quad \quad \quad \quad 
\begin{tikzpicture}[style=thick,scale=0.65]

\node[draw,circle,inner sep=0.1cm] (1) at (0,3) [thick] {\textbf{1}};
\node[draw,circle,inner sep=0.1cm] (1') at (2,3) [ thick] {\textbf{1}};
\node[draw,circle,inner sep=0.1cm] (2) at (0,2) [ thick] {\textbf{2}};
\node[draw,circle,inner sep=0.1cm] (2') at (2,2) [ thick] {\textbf{2}};
\node[draw,circle,inner sep=0.1cm] (3) at (0,1) [ thick] {\textbf{3}};
\node[draw,circle,inner sep=0.1cm] (3') at (2,1) [ thick] {\textbf{3}};
\node[draw,circle,inner sep=0.1cm] (4) at (2,0) [thick] {\textbf{4}};
\node[draw,circle,inner sep=0.1cm] (5) at (2,-1) [ thick] {\textbf{5}};
\node[draw,circle,inner sep=0.1cm] (6) at (2,-2) [ thick] {\textbf{6}};

\draw[red,line width = 3] (1) -- (1');
\draw[red,line width = 3] (2) -- (2');
\draw[red,line width = 3] (3) -- (3');
\draw[blue,dotted,line width = 3] (1) -- (4);
\draw[blue,dotted,line width = 3] (2) -- (5);
\draw[blue,dotted,line width = 3] (3) -- (6);
\end{tikzpicture}
\caption{$(2,3,6)$ on the left has type $(\varnothing | 3)$, $(2,1,5)$ has type $(2 | 1)$, and $(4,5,6)$ has type $(\varnothing | 1^3)$.}
\label{fig:graphs}
\end{figure}
\end{center}
Recall that $|\lambda|$ denotes the size of the integer partition $\lambda$, i.e., the number of cells in its Young diagram, and $\len{\lambda}$ denotes the number of parts of $\lambda$, i.e., the number of rows in its Young diagram. Let $\classes_n := \left\{(\lambda|\rho)\colon|\lambda|+|\rho|=n\right\}$ where $\lambda$ and $\rho$ are partitions. When $m\ge 2n$, $\classes_n$ is the set of all cycle-path types.
Note that $(\varnothing|1^n)$ is not a cycle-path type when $m<2n$, and for $m=n$, all cycle-path types are of form $(\lambda|\varnothing)$, where $\lambda\vdash n$.
We can decompose $\classes_n$ as a disjoint union $\classes_n=\bigcup_{k=0}^n\classes_{n,k}$, where $\classes_{n,k}$ consists of all $(\lambda|\rho)\in\classes_n$ such that $\len{\rho}=n-k$.
 {Note that, for any two $f,f'\in S_{n,m}$, having $\type{f,f'}\in \classes_{n,k}$ implies $|\text{im}~f\cap \text{im}~f'|=k$.}

Recall that any irrep in $\text{Irr}(S_{n,m})$ is of the form $\lambda \otimes \lambda'$ where $\lambda' / \lambda$ is a horizontal strip of size $m-n$. To see that cycle-path types $(\tau | \rho)$ have a natural correspondence with these irreducibles, consider a Young diagram of $\lambda'$ such that the cells of $\lambda' / \lambda$ are marked.
Every columns of $\lambda$ in $\lambda'$ with a marked cell below it corresponds to a part in $\rho$ whereas an unmarked column correspond to a part in $\tau$. 
For instance, taking $\lambda = (2,1)$ and $m = 7$, we have
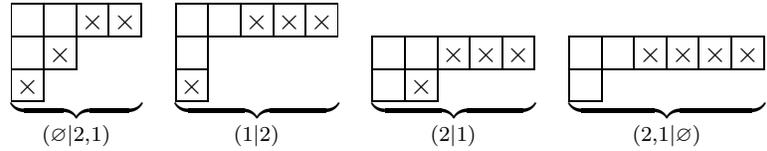
\[\underbrace{\young(~~\times\times,~\times,\times)}_{(\varnothing | 2,1)} \quad 
\underbrace{\young(~~\times\times\times,~,\times)}_{(1|2)} \quad 
\underbrace{\young(~~\times\times\times,~\times)}_{(2|1)} \quad  
\underbrace{\young(~~\times\times\times\times,~)}_{(2,1|\varnothing)}. \]
Note that the marked singleton columns correspond to paths of length zero (i.e., isolated nodes). For each cycle-path type $(\tau | \rho)$, the \emph{$(\tau | \rho)$-associate} of $\mathcal{A}_{n,m}$ is the following $\ffac{m}{n} \times \ffac{m}{n}$ binary matrix:
\[
    (A_{(\tau|\rho)})_{i,j} = 
\begin{cases}
    1,& \text{if } \type{i,j} = (\tau|\rho)\\
    0,              & \text{otherwise}
\end{cases}
\quad \quad \quad  \text{for all $i,j \in S_{m,n}$.} 
\]

\subsection{The Valencies and Multiplicities}

For each $0 \leq i \leq d$, let $d_i := \tr E_i$ denote the \emph{multiplicity} of the $i$th eigenspace of an association scheme, that is, the dimension of its $i$th eigenspace. For each $0 \leq i \leq d$, define the \emph{valency} $v_i$ to be the row sum of an arbitrary row of $A_i$ (equivalently, the largest eigenvalue of $A_i$).  We now give formulas for the valencies $v_{(\lambda | \rho)}$ and multiplicities $d_{(\lambda | \rho)}$ of $\mathcal{A}_{n,m}$.

For each $(\lambda | \rho)$, define the \emph{$(\lambda | \rho)$-sphere} to be the following set:
\[ \Omega_{(\lambda | \rho)} := \{f \in S_{n,m} : \type{f} = (\lambda|\rho)\}.\]
The spheres partition $S_{n,m}$ and it useful to think of them as conjugacy classes.  Indeed, when $n=m$, these spheres are the conjugacy classes of $S_m$. Note that $v_{(\lambda | \rho)} = |\Omega_{(\lambda | \rho)}|$, and basic combinatorial reasoning reveals the following.
\begin{prop}\label{prop:spheresizes}
 For any cycle-path type $(\lambda | \rho)$, the size of the $(\lambda|\rho)$-sphere is
\[ v_{(\lambda | \rho)} =  |\Omega_{(\lambda | \rho)}| = \frac{n!}{\prod_{i =1}^n i^{\ell_i} \ell_i!r_i! } \fFac{m-n}{\len{\rho}}\]
where $\lambda=(n^{\ell_n},\cdots,1^{\ell_1})$, $\rho=(n^{r_n},\cdots,1^{r_1})$, and $\len{\rho}=r_1+\cdots+r_n$.
\end{prop} 

The multiplicities $d_{(\tau | \rho)}$ are easy to deduce due to the fact that each eigenspace of the scheme is isomorphic to an irrep $\mu \otimes \lambda$ of $S_n \times S_m$, and that $\dim \mu \otimes \lambda = \dim \mu \cdot \dim \lambda$.  As we have seen, these dimensions are counted by the hook rule.  In particular, for a cycle-path type $(\tau|\rho)$, let $\tau\cup\rho$ be the union of the set of parts of the two partitions. Then we have $d_{(\tau | \rho)} = d_{\lambda \otimes \lambda'}$ such that $\lambda=(\tau\cup\rho)^\top\vdash n$, $\lambda'=(\tau\cup(m-n,\rho^\top)^\top)^\top$, and `$\top$' denotes the transpose partition.

\providecommand{\id}{\mathsf{id}}
\subsection{Stirling Numbers and Path Covers}

The proof of the main result will rely on some combinatorial estimates of sizes of certain unions of spheres. These sizes are closely related to the (unsigned) Stirling numbers of the first kind.
Recall that for any positive integers $n,k$, the \emph{(unsigned) Stirling numbers of the first kind} are defined by the following recurrence:
\[
{0 \brack 0} = 1; \quad {n \brack 0} = {0 \brack n} = 0; \quad {n+1 \brack k} = n{n \brack k} + {n \brack k-1}. 
\]
A \emph{directed cycle cover} of a directed graph is a union of directed cycles that partition the vertices of the graph. It is well-known that ${n \brack k}$ counts the number of directed cycle covers of the complete directed graph $\vec{K}_n := ([n],[n] \times [n])$ on $n$ vertices and $n^2$ arcs that have precisely $k$ cycles. The following upper bound is also well-known:
\begin{equation}\label{eq:stirling}
{n \brack n-k}\leq\frac{n^{2k}}{2^kk!}.
\end{equation}
A \emph{directed path cover} of a directed graph is a union of directed paths that partition the vertices of the graph.
Let ${n \brack k}'$ denote the number of directed path covers of $\vec{K}_n$ that have precisely $k$ paths, where an isolated vertex is considered a trivial path. These numbers are given by the following recurrence:
\[
{0 \brack 0}' = 1; \quad {n \brack 0}' = {0 \brack n}' = 0; \quad {n+1 \brack k}' = (n+k){n \brack k}' + {n \brack k-1}'. 
\]
Indeed, there are $(n+k)$ ways of extending any given directed path cover of $\vec{K}_{n}$ with $k$ paths to a directed path cover of $\vec{K}_{n+1}$ with $k$ paths.
It is not hard to show an upper bound akin to (\ref{eq:stirling}).
\begin{prop}\label{clm:bounds} For all $k \leq n$ we have
	$$ {n \brack n - k}' \leq \frac{(2n-k)^{2k}}{2^kk!} \leq  2^k\left( \frac{n^{2k}}{ k!} \right).$$
\end{prop}
\begin{proof} By induction, we have
	\begin{align*}
	{n \brack n-k}'&= (2n-k-1){n-1 \brack n-k}' + {n-1 \brack n-k-1}'\\
	&= (2n-k-1){n-1 \brack (n-1)-(k-1)}' + {n-1 \brack (n-1)-k}'\\
	&\leq (2n-k-1) \frac{(2n-k-1)^{2(k-1)}}{2^{k-1}(k-1)!} + \frac{(2n-k-2)^{2k}}{2^k k!}\\
	&\leq (2n-k-1) \frac{(2n-k-1)^{2(k-1)}}{2^{k-1}(k-1)!} + \frac{(2n-k-1)^{2k}}{2^k k!}\\
	&= \frac{1}{2^{k}k!} (2k + (2n - k -1)) (2n-k-1)^{2k-1}\\
	&\overset{\text{AM-GM}}{\leq} \frac{1}{2^{k}k!} \left(\frac{(2k + (2n-k -1)) + (2k-1)(2n-k-1)}{2k}\right)^{2k}\\
	&= \frac{(2n-k)^{2k}}{2^k k!} \\
	&\leq 2^k \left( \frac{n^{2k}}{k!} \right).
	\end{align*}
\end{proof}

\noindent For the remainder of this section, assume that $m \geq n^{2 + \alpha}$ for some $\alpha > 0$. For any injection $f \in S_{n,m}$, define the \emph{path support} of $f$ to be the subset of vertices of $[n]$ that belong to a path of $f \cup f_{\id}$. Let $S_{n,m,k} \subseteq S_{n,m}$ be the set of injections $f \in S_{n,m}$ such that $\text{d}(f) \in \classes_{n,k}$, i.e., $f \cup f_{\id}$ has exactly $n-k$ non-trivial paths. Let $S_{n,m,k,j} \subseteq S_{n,m,k}$ be the set of injections $f \in S_{n,m,k}$ with a path support of size $n-j$. Define the probabilities 
$$p_k := p_{n,m,k} = |S_{n,m,k}|/\ffac{m}{n} \quad \text{ and } \quad p_{k,j} := p_{n,m,k,j} = |S_{n,m,k,j}|/\ffac{m}{n},$$
so that $\sum_{j=0}^k p_{k,j} = p_k$ and $\sum_{k=0}^n p_k = 1$. In what follows, we have $j \leq k \leq n$, and we shall think of these probabilities $p_k$ and $p_{k,j}$ as being functions of $n$. Basic combinatorial reasoning shows that 
\[
|S_{n,m,k,j}| = \binom{n}{j}j!~{n-j \brack n-k}'\ffac{(m-n)}{n-k}, ~\text{ thus }~ p_{k,j} = \frac{{n-j \brack n-k}' \ffac{n}{j}~\ffac{(m-n)}{n-k}}{\ffac{m}{k} \cdot \ffac{(m-k)}{n-k}} \leq {n-j \brack n-k}'  \frac{\ffac{n}{j}}{\ffac{m}{k}}. 
\]
By Proposition~\ref{clm:bounds}, we have
$$ p_{k,j} \leq {n-j \brack n-k}' \frac{\ffac{n}{j}}{\ffac{m}{k}}  = {(n-j) \brack (n-j) - (k-j)}' \frac{\ffac{n}{j}}{\ffac{m}{k}}
\leq 2^{k-j}\frac{(n-j)^{2(k-j)}}{(k-j)!}  \frac{\ffac{n}{j}}{\ffac{m}{k}} \leq 4 \left( \frac{(n^2)^{k-j/2}}{(n^{2+\alpha}-n)^{k}} \right) .
$$
For sufficiently large $n$, it is immediate for any positive integer $k$ that $p_{k,j}$ is maximized when $j=0$. In particular, we have the following proposition. 
\begin{prop}\label{prop:small}
	Let $t$ be a positive integer. For all $k \geq t$ and $0 \leq j \leq k$, we have $p_{k,j} =  \mathrm{O}\left(  1/n^{t\alpha }  \right)$.
\end{prop}

\subsection{The Johnson Ordering of $\mathcal{A}_{n,m}$}

The \emph{Johnson scheme} $\mathcal{J}(m,n)$ is a symmetric association scheme defined over the $n$-subsets of $[m]$. The $i$th associate $A_i \in \mathcal{J}(m,n)$ of the Johnson scheme is defined such that $(A_i)_{X,Y} = 1$ if $n-|X \cap Y| =i$, and is 0 otherwise for any two $n$-subsets $X,Y$. It is well-known that the $i$th eigenspace of $\mathcal{J}(m,n)$ is isomorphic to the $S_m$-irrep associated to the partition $(m-i,i) \vdash m$. For proofs of these facts and more, see~\cite{GodsilMeagher}.
Henceforth, let $E_{i}$ be the primitive idempotent of the Johnson scheme that projects onto $\cV_{(m-i,i)}$.

There exists a natural ordering of the $S_{n,m}$ that we call \emph{the Johnson ordering} that shows the Johnson scheme is a quotient of $\mathcal{A}_{n,m}$.
First, we order $S_{n,m}$ by the corresponding $n$-subsets (the particular order does not matter). Next, we lexicographically order all $n!$ injections that map to the same $n$-subset (i.e., share the same image), e.g., for $n=3$, we have:
\begin{align*}
&(1,2,3),(1,3,2),(2,1,3),(2,3,1),(3,1,2),(3,2,1),(1,2,4),(1,4,2),(2,1,4),(2,4,1),(4,1,2),\cdots.
\end{align*}
Both $S_n$ and $S_m$ act on the domain and range of an injection respectively, and so each of their actions correspond to some collection of $ m^{\underline{n}}\times  m^{\underline{n}}$ permutation matrices (i.e., their corresponding permutation representations). The action of $S_m$ on $S_{m,n}$ is transitive, but $S_n$'s action has $\binom{m}{n}$ orbits, one for each $n$-subset. Note that on all $n!$ permutations of $S_n$ corresponding to any given $n$-subset, the action of $S_n$ corresponds to the regular representation of $S_n$.

Given $\lambda\vdash n$ and $\lambda'\vdash m$, let $E_\lambda$ and $E_{\lambda'}$ be the orthogonal projectors on the $\lambda$-isotypic and $\lambda'$-isotypic subspaces, respectively. Since the actions of $S_n$ and $S_m$ on $S_{n,m}$ commute, $E_\lambda$ and $E_{\lambda'}$ also commute, and we have $E_{\lambda\otimes\lambda'}=E_{\lambda}E_{\lambda'}$. From the specific way we ordered $S_{n,m}$ in the previous paragraph, for any $\lambda\vdash n$ we have
\[
E_\lambda=I_{\binom{m}{n}}\otimes F_\lambda = \underbrace{F_\lambda\oplus F_\lambda\oplus\cdots\oplus F_\lambda}_{\binom{m}{n}\text{ times}},
\]
where $F_\lambda$ is the $n!\times n!$ orthogonal projector on the $\lambda$-isotypic subspace of the regular representation of $S_n$.
Hence, we can write $E_{\lambda\otimes\lambda'}$ as a product of two block matrices:
\[
E_{\lambda\otimes\lambda'}=\left(
\begin{array}{ccc}
B_{1,1}^{\lambda'} & B_{1,2}^{\lambda'} & \cdots\\
 B_{2,1}^{\lambda'} & B_{2,2}^{\lambda'} \\
 \vdots & & \ddots
\end{array}
\right)
\left(
\begin{array}{ccc}
F_\lambda & 0 & \cdots\\
 0 & F_\lambda \\
 \vdots & & \ddots
\end{array}
\right)
=
\left(
\begin{array}{ccc}
B_{1,1}^{\lambda'}F_\lambda & B_{1,2}^{\lambda'}F_\lambda & \cdots\\
 B_{2,1}^{\lambda'}F_\lambda & B_{2,2}^{\lambda'}F_\lambda \\
 \vdots & & \ddots
\end{array}
\right)
,
\]
where the first matrix is $E_{\lambda'}$, in which each block $B_{i,j}^{\lambda'}$ is some $n!\times n!$ matrix. 
 {For every $\lambda\vdash n$, we have $E_{\lambda\otimes\lambda'}=E_{\lambda}E_{\lambda'}E_{\lambda}$, which means that, if the rank of $F_\lambda$ is $1$, then $E_{\lambda\otimes\lambda'}$ can be expressed as a tensor product with one of the factors being $F_\lambda$. When $\lambda=(n)$, we have $F_\lambda=F_{(n)}=J/n!$,} 
where $J$ is the $n!\times n!$ all-ones matrix. Thus, from the expression above, we have
\begin{align}\label{eq:johnson2}
E_{(n)\otimes(m-1,1)} = J/n! \otimes E_1 = \left(
\begin{array}{ccc}
b_{1,1}J & b_{1,2}J & \cdots\\
b_{2,1}J & b_{2,2}J \\
 \vdots & & \ddots
\end{array}
\right),
\end{align}
where $b_{i,j}$ are scalars, $1/\ffac{m}{n}$ times the dual eigenvalues $q_1(\cdot)$ (see Lemma~\ref{lem:dualeig} for explicit expressions).
Thus we have

\begin{align}\label{eq:johnson}
E_{(n)\otimes(m-1,1)}\circ E_{\lambda\otimes\lambda'} =
\left(
\begin{array}{ccc}
b_{1,1}B_{1,1}^{\lambda'}F_\lambda & b_{1,2}B_{1,2}^{\lambda'}F_\lambda & \cdots\\
b_{2,1}B_{2,1}^{\lambda'}F_\lambda & b_{2,2}B_{2,2}^{\lambda'}F_\lambda \\
 \vdots & & \ddots
\end{array}
\right),
\end{align}
which is orthogonal to $E_{\mu}$ (and thus $E_{\mu \otimes \mu'}$) for all $\mu \vdash n$ such that $\mu \neq\lambda$  {because of $F_\lambda F_\mu=0$}.

\subsection{The Dual Eigenvalues of $\mathcal{A}_{n,m}$}

It is well-known that the primitive idempotents $E_i$ of an association scheme can be written as a unique linear combination of associates $A_i$ of the scheme (see~\cite[Ch. 2.1]{GodsilAssoc}):
\begin{align*}
E_{i} &= \frac{1}{|X|} \sum_{j=0}^d q_{i}(j)A_{j}. 
\end{align*}
The $q_i(j)$'s are called the \emph{dual eigenvalues} of the scheme. In our case, this specializes to
\begin{align*}
E_{\lambda \otimes \lambda'} &= \frac{1}{m^{\underline{n}}} \sum_{(\mu|\rho)\in\classes_n} q_{\lambda \otimes \lambda'}( \mu | \rho  )A_{( \mu | \rho )}, 
\end{align*}
and these coefficients $q_{\lambda \otimes \lambda'}( \mu | \rho  )$ are the \emph{dual eigenvalues} of $\mathcal{A}_{m,n}$. 

\begin{prop}\label{prop:idem}\emph{\cite{CST}}
	For any finite symmetric Gelfand pair $(G,K)$, the dual eigenvalues $q_i(j)$ of the symmetric association scheme on $X = G/K$ can be written as
	\[
	q_i(j) = d_i\omega^i_j
	\]
	where $d_i$ is the dimension of irreducible $i$ corresponding to the spherical function $\omega^i$.
\end{prop}
Let us consider matrices in the Bose--Mesner algebra $\mathfrak{A}_{n,m}$ of the injection association scheme. By symmetry, every such matrix can be specified by a row or column corresponding to a single injection.  Note that
\[ \left(E_{\lambda \otimes \lambda'}\right)_{f,h} = \frac{q_{\lambda \otimes \lambda'}( \type{f,h}  )}{\ffac{m}{n}} \]
and that 
$$\sum_{f \in S_{n,m}} q_{\lambda \otimes \lambda'}( \type{f,h} )~1_f = \ffac{m}{n}\left(E_{\lambda \otimes \lambda'}\right)_{h} \in \lambda \otimes \lambda'$$
for all $\lambda \otimes \lambda'$ and $f,h \in S_{n,m}$,
where  $1_f\in\C[S_{n,m}]$ denotes the binary unit vector with the unique $1$ in position $f$.
It is well-known that the projector $E_{\lambda'}$ onto the $\lambda'$-isotypic component can be written as
\[ \left(E_{\lambda'}\right)_{f,h} = \frac{d_{\lambda'}}{m!} \sum_{\sigma\in S_m} \chi_{\lambda'}(\sigma^{-1})(V_\sigma)_{f,h}\]
where $V_\sigma:1_f\mapsto 1_{\sigma*f}$ for all $f,h \in S_{n,m}$
 {and $\chi_{\lambda'}$ is the character corresponding to $\lambda'$.}
 The foregoing, and the fact that $E_{\lambda'}E_{\lambda \otimes \lambda'} = E_{\lambda \otimes \lambda'}$ implies the following proposition.
\begin{prop}\label{prop:charSum}
For any $f,h\in S_{n,m}$, $\lambda\vdash n$, and $\lambda'\vdash m$, we have
\[
q_{\lambda\otimes\lambda'}\big(\type{f,h}\big) = \frac{d_{\lambda'}}{m!}
\sum_{\sigma\in\Sym_m} {\chi_{\lambda'}(\sigma)}q_{\lambda\otimes\lambda'}\big(\type{\sigma^{-1}*f,h}\big).
\]
\end{prop}
\begin{proof}
 {
We have $\chi_{\lambda'}(\sigma^{-1})=\chi_{\lambda'}(\sigma)$.
By applying $E_{\lambda'}$ to $\ffac{m}{n}\left(E_{\lambda \otimes \lambda'}\right)_{h}$, we get
\begin{align*}
\ffac{m}{n}\left(E_{\lambda'}E_{\lambda \otimes \lambda'}\right)_{h}
 & = \sum_{f \in S_{n,m}} \frac{d_{\lambda'}}{m!} \sum_{\sigma\in S_m} \chi_{\lambda'}(\sigma) q_{\lambda \otimes \lambda'}( \type{f,h} )~1_{\sigma * f}
 \\ & = \sum_{f \in S_{n,m}} \frac{d_{\lambda'}}{m!} \sum_{\sigma\in S_m} \chi_{\lambda'}(\sigma) q_{\lambda \otimes \lambda'}( \type{\sigma^{-1}*f,h} )~1_{f}.
\end{align*}
The proposition follows by equating the coefficients of $1_{f}$ in the expressions for $\ffac{m}{n}\left(E_{\lambda \otimes \lambda'}\right)_{h}$ and $\ffac{m}{n}\left(E_{\lambda'}E_{\lambda \otimes \lambda'}\right)_{h}$.
}
\end{proof}

\begin{lem}\label{lem:dualeig}
	Let $q_i(j)$ be a dual eigenvalue of the Johnson scheme $\mathcal{J}(m,n)$. Then we have
	\[  q_1(j) = \frac{\binom{m}{n}}{\binom{m-2}{n-1}} \left(  n- j - \frac{n^2}{m}   \right). \]
	Moreover, if $m \geq n^2$, then $q_1(j) \geq 0$ for all $j \neq n$.
\end{lem}
\begin{proof}
	Let $p_{i}(j)$ denote the $j$-th eigenvalue of the $i$-th associate of the Johnson scheme $\mathcal{J}(m,n)$. It is well-known (see~\cite{GodsilMeagher} for example) that
	\[ p_{i}(j) = \sum_{r=i}^n (-1)^{(r-i+j)} \binom{r}{i} \binom{m-2r}{n-r} \binom{m-r-j}{r-j}. \]
	 {Using basic relations between primal and dual eigenvalues (see~\cite{GodsilAssoc}), we can write the first primitive idempotent $E_1$ of the Johnson scheme as follows:
	\begin{align*}
	E_1 &= \frac{1}{\binom{m}{n}}\sum_{j=0}^n q_1(j) A_j\\ 
	&= \frac{m-1}{\binom{m}{n}} \sum_{j=0}^n\frac{ p_j(1)}{\binom{n}{j}\binom{m-n}{j}} A_j\\
	&=\frac{m-1}{\binom{m}{n}}\sum_{j=0}^n \frac{ 1 }{\binom{n}{j}\binom{m-n}{j} } \left(\binom{n-1}{j} \binom{m-n-1}{j} -
	\binom{n-1}{j-1} \binom{m-n-1}{j-1} \right)A_j\\
	&=\frac{m-1}{\binom{m}{n}} \sum_{j=0}^n \left( 1 -\frac{mj}{n(m-n) } \right)A_j\\
	&=\frac{1}{\binom{m-2}{n-1}} \sum_{j=0}^n \left(n-j - \frac{n^2}{m} \right)A_j.
	\end{align*}
	Equating coefficients of $A_j$ gives the result.}
%
%
\end{proof}
\begin{lem}\label{lem:dualeig2}
	For all $\mu\in\classes_{n,k}$, we have $$q_{(n)\otimes(m-1,1)}(\mu)=\frac{(k m-n^2)(m-1)}{n(m-n)}.$$
\end{lem}
\begin{proof}
	Recall that the $i$th eigenspace of the Johnson scheme $\mathcal{J}(m,n)$ is isomorphic to the $S_m$-irrep associated to the partition $(m-i,i) \vdash m$, and that $E_{i}$ denotes the primitive idempotent of the Johnson scheme that projects onto its $(m-i,i)$ eigenspace. 
	 {In the proof of the previous lemma we saw that}
	\begin{align*}
	E_{1} &= \frac{1}{\binom{m}{n}} \sum_{j=0}^n \left[   \frac{\binom{m}{n}}{\binom{m-2}{n-1}} \left(  n- j - \frac{n^2}{m}   \right) \right]A_j = \frac{1}{\binom{m-2}{n-1}} \sum_{j=0}^n \left(  n- j - \frac{n^2}{m}   \right) A_j.
	\end{align*}
	Since $(A_j)_{X,Y} = 1$ only if $|X \cap Y| = n - j$, we have the dual eigenvalue
	\[\left(\ffac{m}{n}E_{(n)\otimes(m-1,1)} \right)_{f,h}
	=  \ffac{m}{n}\frac{|\text{im}~f\cap \text{im}~h| - n^2/m}{n! \binom{m-2}{n-1}}
	=\frac{(m-1)(|\text{im}~f \cap \text{im}~h|m - n^2)}{n(m-n)}.
	\]
 {	
This value is clearly the same for all pairs of $f$ and $h$ that have the same $|\text{im}~f\cap \text{im}~h|$, in other words, for all pairs of $f$ and $h$ for which the cycle-path type $\type{f,h}$ is in the same $\classes_{n,k}$. Therefore, it}	
	follows that
	$$q_{(n)\otimes(m-1,1)}(\mu)=\frac{(k m-n^2)(m-1)}{n(m-n)}$$ 
	for all $\mu\in\classes_{n,k}$, which completes the proof.
\end{proof}

\section{A sufficient condition on Krein parameters} \label{sec:AdvToKrein}
In this section, we reduce the proof of Theorem~\ref{thm:main} to an upper bound on certain Krein parameters of $\mathcal{A}_{n,m}$. Recall that $\circ$ denotes the Schur (entrywise) product of two matrices.
\begin{definition}[Krein Parameters]
	Let $\mathcal{A}$ be an association scheme on $v$ vertices with $d$ associates. For any $0 \leq i,j \leq d$, there exist constants $q_{i,j}(k)$ such that 
	\[E_i \circ E_j = \frac{1}{v} \sum_{k=0}^d q_{i,j}(k) E_k,\]
	which are called the \emph{Krein parameters} of $\mathcal{A}$. More explicitly, we have 
	\[q_{i,j}(k) = v \frac{\Tr{E_k(E_i \circ E_j)}}{d_k}.\]
\end{definition}
\noindent The Krein parameters can alternatively be written as
\begin{equation}\label{eq:Krein}
q_{i,j}(k)  = \frac{1}{vd_k} \sum_{\ell=0}^d \frac{q_i(\ell)q_j(\ell) \overline{q_k(\ell)}}{v_\ell} =  \frac{d_id_j}{v} \sum_{\ell=0}^d \frac{\overline{p_i(\ell)} \overline{p_j(\ell)} p_k(\ell)}{v_\ell^2}, 
\end{equation}
where $p_i(j)$ denotes the $j$-th eigenvalue of $A_i$ and $q_i(j)$ denotes the $j$th dual eigenvalue of $E_i$ (see~\cite[Chap. 2.4]{GodsilAssoc} for a proof).

To prove the lower bound on non-coherent \textsc{Index Erasure}, we use the same adversary matrix $\Gamma$ as \cite{AMRR11} used for the coherent case.%
\footnote{Technically, the adversary matrix used here is $\sqrt{n}$ times that of \cite{AMRR11} as the adversary method they use places slightly different conditions on the adversary matrix.}
 For simplifying the equations, without loss of generality let us assume that $n$ is a square. As in \cite{AMRR11}, we choose
\[
\Gamma := \sum_{k=0}^{\sqrt{n}-1} (\sqrt{n}-k)\sum_{\lambda\vdash k}E_{(n-k,\lambda)\otimes(m-k,\lambda)},
\]
and thus the orthogonal projection onto its image is
$$\Pi_\Gamma := \sum_{\lambda\colon|\lambda| < \sqrt{n}}E_{(n-|\lambda|,\lambda)\otimes(m-|\lambda|,\lambda)}.$$

Note that the sole principal eigenvector $\omega$ of $\Gamma$ is the uniform superposition over $\Dom$ (i.e., $\omega_f=1/\sqrt{\ffac{m}{n}}$ for all $f\in\Dom$).
Thus, as per Corollary~\ref{cor:adv}, we are interested in the quantity
\[
\eta = \max_{T\in\junkSet} \Tr{\Pi_\Gamma\frac{\left(T\circ T^\odot\right)}{\ffac{m}{n}}}.
\]
 {As described by the automorphism principle (Theorem~\ref{thm:auto}), here it suffices to consider $T$ that are $(S_n\times S_m)$-invariant, that is, $T$ that belong the Bose--Mesner algebra $\mathfrak{A}_{n,m}$. Because of that, for simplicity, let us redefine to $\junkSet$ be the set of all state matrices in $\mathfrak{A}_{n,m}$.}

For any primitive idempotent $E_{\lambda \otimes \lambda'}$, let 
\[ T_{\lambda\otimes\lambda'} := \left( \frac{\ffac{m}{n}}{\tr E_{\lambda \otimes \lambda'}} \right) E_{\lambda \otimes \lambda'} = \left( \frac{\ffac{m}{n}}{d_{\lambda \otimes \lambda'}} \right) E_{\lambda \otimes \lambda'}\]
be its associated state matrix. In~\cite{AMRR11} it is shown that the target matrix can be written as
\[
T^\odot=\frac{n}{m}T_{(n)\otimes(m)}
+ \left(1-\frac{n}{m}\right)T_{(n)\otimes(m-1,1)}.
\]
In the coherent case, recall that $T=J$, and therefore 
\[
\eta
=\Tr{\Pi_\Gamma \frac{T^\odot}{\ffac{m}{n}}}
=\frac{n}{m}\Tr{\frac{T_{(n)\otimes(m)}}{\ffac{m}{n}}}
=\frac{n}{m}.
\]
The most technically involved part of the proof of the lower bound by \cite{AMRR11} is proving that $\|\Delta_x\circ\Gamma\| =  \OO(1)$.
Since we are using the same adversary matrix $\Gamma$, we already have the above bound on
$\|\Delta_x\circ\Gamma\|$.
Our goal is to show that
$\Tr{\Pi_\Gamma(T\circ T^\odot)/\ffac{m}{n}}$ is small for \emph{all} state matrices $T$.

By dividing the elements in the set $\junkSet$ by $\ffac{m}{n}$ we obtain the set of all density matrices (positive-semidefinite Hermitian matrices with trace 1) of the Bose--Mesner algebra $\mathfrak{A}_{n,m}$. Observe that $(T \circ T')/\ffac{m}{n}$ is a density matrix for all $T,T' \in \junkSet$. For any $\junk\in\junkSet$, we have
\[
\Tr{\Pi_\Gamma\frac{\left(T\circ T_{(n)\otimes(m)}\right)}{\ffac{m}{n}}}
= \Tr{\Pi_\Gamma\frac{T}{\ffac{m}{n}}} \le 1,
\]
therefore
\[
\Tr{\Pi_\Gamma\frac{\left(T\circ T^\odot\right)}{\ffac{m}{n}}}
\le \frac{n}{m} + \left(1-\frac{n}{m}\right)\Tr{\Pi_\Gamma\frac{\left(T\circ T_{(n)\otimes(m-1,1)}\right)}{\ffac{m}{n}}}.
\]
Our goal is to bound the latter term:
\begin{align}\label{eq:termToBound}
\left(1-\frac{n}{m}\right)\Tr{\Pi_\Gamma\frac{\left(T\circ T_{(n)\otimes(m-1,1)}\right)}{\ffac{m}{n}}}.
\end{align}
Note that
$
\junkSet = \left\{\sum_\chi c_\lambda\junk_\chi
\colon \sum_\chi c_\chi=1 \;,  \; c_\chi\ge 0 \right\},
$
where the sums range over $\chi\in\text{Irr}(S_{n,m})$. Hence,
\begin{align*}
\max_{\junk\in\junkSet}\Tr{\Pi_\Gamma\frac{\left(\junk\circ T_{(n)\otimes(m-1,1)}\right)}{\ffac{m}{n}}}
& =
\max_{\substack{\{c_\chi\ge 0\}_\chi\\\sum_\chi c_\chi=1}} c_\chi\Tr{\Pi_\Gamma\frac{\left(\junk_\chi\circ T_{(n)\otimes(m-1,1)}\right)}{\ffac{m}{n}}}
\\ & =
\max_{\chi}\Tr{\Pi_\Gamma\frac{\left(\junk_{\chi}\circ T_{(n)\otimes(m-1,1)}\right)}{\ffac{m}{n}}}.
\end{align*}
The following proposition simplifies~(\ref{eq:termToBound}).
\begin{prop} For any $\lambda=(n-|\nu|,\nu)$ and $\bar\lambda = (m-|\nu|,\nu)$, we have 
	\[
		\Tr{\Pi_\Gamma\frac{\left(\junk_{\lambda\otimes\lambda'}\circ T_{(n)\otimes(m-1,1)}\right)}{\ffac{m}{n}}}
		=\Tr{E_{\lambda\otimes\bar\lambda}\frac{\left(\junk_{\lambda\otimes\lambda'}\circ T_{(n)\otimes(m-1,1)}\right)}{\ffac{m}{n}}}.
	\]
	Moreover, if $|\nu| \geq \sqrt{n}$, then
	$
		\Tr{\Pi_\Gamma\frac{\left(\junk_{\lambda\otimes\lambda'}\circ T_{(n)\otimes(m-1,1)}\right)}{\ffac{m}{n}}} = 0.
	$
\end{prop}
\begin{proof}
By Equation~(\ref{eq:johnson}), if $|\nu| < \sqrt{n}$, then we have
\begin{align*}
\Pi_\Gamma \frac{\left(\junk_{\lambda\otimes\lambda'}\circ T_{(n)\otimes(m-1,1)}\right)}{\ffac{m}{n}}
&= \sum_{\mu \colon|\mu| < \sqrt{n}}E_{(n-|\mu|,\mu)\otimes(m-|\mu|,\mu)} \frac{\left(\junk_{\lambda\otimes\lambda'}\circ T_{(n)\otimes(m-1,1)}\right)}{\ffac{m}{n}}\\
&= E_{(n-|\nu|,\nu)\otimes(m-|\nu|,\nu)} \frac{\left(\junk_{\lambda\otimes\lambda'}\circ T_{(n)\otimes(m-1,1)}\right)}{\ffac{m}{n}}\\
&= E_{\lambda \otimes \bar \lambda} \frac{\left(\junk_{\lambda\otimes\lambda'}\circ T_{(n)\otimes(m-1,1)}\right)}{\ffac{m}{n}}.
\end{align*}
The second part of the proof is immediate from the definition of $\Gamma$.
\end{proof}
\noindent The proposition above now allows us to bound (\ref{eq:termToBound}) for all $\lambda \vdash n$ such that $|\lambda|-\lambda_1 \leq \sqrt{n}$.
\begin{cor}\label{cor:lambdabar}
Suppose $\lambda\vdash n$ has no more than $\sqrt{n}$ cells below the first row. Then for all $\lambda'\neq\bar\lambda$ we have
\[
\tr\left[
E_{\lambda \otimes \overline{\lambda}}
 \frac{\left(T_{(n) \otimes (m-1,1)}\circ T_{\lambda \otimes \lambda'}\right)}{\ffac{m}{n}} 
\right]
\in \OO(\sqrt{n}/m).
\]
\end{cor}
\begin{proof}
Let $\mathrm{sum}[\cdot]$ denote the sum of the entries of the matrix. 
 {We have $T_{\lambda \otimes \lambda''}/\ffac{m}{n}=E_{\lambda \otimes \lambda''}/d_{\lambda\otimes\lambda''}$ for all $\lambda''\vdash m$. Therefore}
\begin{align*}
\tr\left[
E_{\lambda \otimes \overline{\lambda}}
 \frac{\left(T_{(n) \otimes (m-1,1)}\circ T_{\lambda \otimes \lambda'}\right)}{\ffac{m}{n}} 
\right]
&= 
\frac{1}{d_{\lambda\otimes\lambda'}} \text{sum}\left[
 E_{\lambda \otimes\overline{\lambda}}\circ T_{(n) \otimes (m-1,1)} \circ E_{\lambda \otimes \lambda'}
\right]
\\
&= 
\frac{d_{\lambda\otimes\overline{\lambda}}}{d_{\lambda\otimes\lambda'}}
\tr\left[
E_{\lambda \otimes \lambda'}
 \frac{\left(T_{(n) \otimes (m-1,1)}\circ T_{\lambda \otimes \overline{\lambda}}\right)}{\ffac{m}{n}} 
\right].\\
\intertext{Since $\left(T_{(n) \otimes (m-1,1)}\circ T_{\lambda \otimes \overline{\lambda}}\right)/\ffac{m}{n}$ is a density matrix, this trace is at most $1$, thus}
& \leq \frac{d_{\overline{\lambda}}}{d_{\lambda'}} \in \OO(\sqrt{n}/m),
\end{align*} 
where the asymptotic bound follows from Corollary~\ref{cor:lambdabarPre}, completing the proof.
\end{proof}
\noindent We therefore have
\[
\eta = \OO
\left(\frac{n}{m} + 
\Tr{E_{\lambda\otimes\bar\lambda}\frac{\left(\junk_{\lambda\otimes\bar\lambda}\circ T_{(n)\otimes(m-1,1)}\right)}{\ffac{m}{n}}}
\right),
\]
and it remains to bound the value
\begin{align}
\notag
\Tr{E_{\lambda\otimes\bar\lambda}\frac{\left(\junk_{\lambda\otimes\bar\lambda}\circ T_{(n)\otimes(m-1,1)}\right)}{\ffac{m}{n}}}
& =
\frac{\ffac{m}{n}}{(m-1)d_{\lambda\otimes\bar\lambda}}
\mathrm{sum}\left[E_{(n)\otimes(m-1,1)}\circ E_{\lambda\otimes\bar\lambda}\circ E_{\lambda\otimes\bar\lambda}\right]
\\ & = \label{eq:SumToBound}
\frac{\sum_{\mu\in\classes_n} v_\mu \cdot q_{(n)\otimes(m-1,1)}(\mu)\cdot q^2_{\lambda\otimes\bar\lambda}(\mu)}{\ffac{m}{n}(m-1)d_{\lambda\otimes\bar\lambda}}
\end{align}
\noindent for all $\lambda \vdash n$ with no more than $\sqrt{n}$ cells below the first row. 

Thus to prove Theorem~\ref{thm:main}, it suffices to show for sufficiently large $m$, say $m\ge n^{3+\epsilon}$, that the value of (\ref{eq:SumToBound}), and thus $\eta$, is $\mathrm{O}(1/\sqrt{n})$ for all $\lambda \vdash n$ with no more than $\sqrt{n}$ cells below the first row.  According to the expression (\ref{eq:Krein}) for Krein parameters, (\ref{eq:SumToBound}) equals
\[
\frac{q_{\lambda\otimes\bar\lambda,\,\lambda\otimes\bar\lambda}((n)\otimes(m-1,1))}{d_{\lambda\otimes\bar\lambda}}
= \frac{q_{\lambda\otimes\bar\lambda,\,(n)\otimes(m-1,1)}(\lambda\otimes\bar\lambda)}{m-1},
\]
and therefore the task of bounding (\ref{eq:SumToBound}) is equivalent to bounding the Krein parameters 
$$q_{\lambda\otimes\bar\lambda,\,\lambda\otimes\bar\lambda}((n)\otimes(m-1,1))\quad \text{and}\quad q_{\lambda\otimes\bar\lambda,\,(n)\otimes(m-1,1)}(\lambda\otimes\bar\lambda).$$

\section{Proof of Theorem~\ref{thm:main}}\label{sec:MainProof}
We are now in a position to finish the proof of the main result. In Section~\ref{sec:AdvToKrein} we saw that it suffices to show (\ref{eq:SumToBound}) is $\mathrm{O}(1/\sqrt{n})$ for sufficiently large $m$, which we state as the following claim.

\begin{clm}
	For all $\lambda \vdash n$ with $\ell \leq \sqrt{n}$ cells below the first row, there is a constant $\alpha > 1$ such that 
	\[
	\frac{1}{\ffac{m}{n}} \sum_{(\mu|\rho) \in \mathcal{C}_n}
	\frac{ v_{(\mu|\rho)} \cdot q_{(n)\otimes(m-1,1)}(\mu|\rho)\cdot q_{\lambda\otimes\bar\lambda}^2(\mu|\rho)}
	{d_{(n)\otimes(m-1,1)} \cdot d_{\lambda\otimes\bar\lambda}} = \mathrm{O}(1/\sqrt{n})
	\]
	holds, provided $m \geq n^{2+\alpha}$.
\end{clm}
\noindent We now prove this claim. Throughout the proof, at no loss of generality we shall assume that $m$ and $\sqrt{n}$ are integers to avoid cumbersome notation.
\begin{proof}
	Let $(\mu|\rho) \in \classes_{n,k}$. Since $m \geq n^{2+\alpha}$, Lemma~\ref{lem:dualeig2} implies that $q_{(n)\otimes(m-1,1)}(\mu|\rho) < 0$ if $k=0$, and for $k > 0$ that
	\begin{align*}
	0 \leq q_{(n)\otimes(m-1,1)}(\mu|\rho) &= \frac{(km-n^2)(m-1)}{n(m-n)}
	\leq  \frac{k(m-n)(m-1)}{n(m-n)}
	= \frac{k(m-1)}{n}.
	\end{align*}
	Note that $d_{(n)\otimes(m-1,1)} = (m-1)$. The $\classes_{n,k}$'s partition $\classes_{n}$, so we may rewrite the sum as
	\begin{align*}
	\frac{1}{\ffac{m}{n}} \sum_{(\mu|\rho) \in \mathcal{C}_n}
	\frac{ v_{(\mu|\rho)} \cdot q_{(n)\otimes(m-1,1)}(\mu|\rho)\cdot 	q_{\lambda\otimes\bar\lambda}^2(\mu|\rho)}
	{ d_{(n)\otimes(m-1,1)} \cdot d_{\lambda\otimes\bar\lambda}}  &\leq \frac{1}{\ffac{m}{n}}\sum_{k=1}^n \frac{k}{n} \sum_{(\mu|\rho) \in\classes_{n,k}}
	\frac{ v_{(\mu|\rho)} q_{\lambda\otimes\bar\lambda}^2(\mu|\rho)}
	{  d_{\lambda}d_{\bar\lambda}}.
	\intertext{By Proposition~\ref{prop:idem}, we can write the dual eigenvalue as}
	&= 
	\frac{1}{\ffac{m}{n}}\sum_{k=1}^n \frac{k}{n} \sum_{(\mu|\rho) \in\classes_{n,k}}
	\frac{ v_{(\mu|\rho)} d_{\lambda\otimes\bar\lambda}^2\left[ \omega^{\lambda\otimes\bar\lambda}_{(\mu|\rho)}\right]^2}
	{  d_{\lambda}d_{\bar\lambda}}\\
	&= 
	\frac{d_{\lambda}d_{\bar\lambda}}{\ffac{m}{n}}\sum_{k=1}^n \frac{k}{n} \sum_{(\mu|\rho) \in\classes_{n,k}}
	v_{(\mu|\rho)} \left[ \omega^{\lambda\otimes\bar\lambda}_{(\mu|\rho)}\right]^2.
	\end{align*}
	\noindent 	For any $f \in S_{n,m}$, define $k_f$ so that $n-k_f$ equals the number of non-trivial paths in the cycle-path type of $f$. From the foregoing, it suffices to show 
	\begin{equation}\label{eq:main}
	\frac{d_{\lambda}d_{\bar{\lambda}}}{\ffac{m}{n}} \sum_{k=1}^n \frac{k}{n} \sum_{(\mu|\rho) \in\classes_{n,k}} v_{(\mu|\rho)}  \left[\omega^{\lambda \otimes \bar\lambda}_{(\mu|\rho)} \right]^2  = \frac{d_{\lambda}d_{\bar{\lambda}}}{\ffac{m}{n}} \sum_{f \in S_{n,m}} \frac{k_f}{n}~\left[  \omega^{\lambda \otimes \bar\lambda}(f) {\omega^{\lambda \otimes \bar\lambda}(f)} \right]  = \mathrm{O}(1/\sqrt{n}).
	\end{equation}
	By Proposition~\ref{prop:orth}, we have 
	$$\frac{d_{\lambda}d_{\bar{\lambda}}}{\ffac{m}{n}} \sum_{f \in S_{n,m}} \mathrm{O}(1/\sqrt{n})~\left[  \omega^{\lambda \otimes \bar\lambda}(f) {\omega^{\lambda \otimes \bar\lambda}(f)} \right] = \mathrm{O}(1/\sqrt{n});$$
	therefore, it suffices to show there exists a constant $c > 0$ such that
	$$
	\frac{d_{\lambda}d_{\bar{\lambda}}}{\ffac{m}{n}} \sum_{\substack{f \in S_{n,m} \\ k_f \geq c\sqrt{n} }}  \omega^{\lambda \otimes \bar\lambda}(f) {\omega^{\lambda \otimes \bar\lambda}(f)} = \mathrm{O}(1/\sqrt{n}).
	$$
	For all $0 \leq j \leq k$, define $\classes_{n,k,j} := \left\{(\mu|\rho) \in \classes_{n,k} : |\mu| = j\right\}$, 
	so that $\classes_{n,k} = \bigsqcup_{j=0}^k \classes_{n,k,j}$.
	Recall that $p_{k,j}$ is the probability that the cycle-path type of $f \in S_{n,m}$ lies in $\classes_{n,k,j}$. We have
	\begin{align*}
	\frac{d_{\lambda}d_{\bar{\lambda}}}{\ffac{m}{n}} \sum_{\substack{f \in S_{n,m} \\ k_f \geq c\sqrt{n} }}  \omega^{\lambda \otimes \bar\lambda}(f) {\omega^{\lambda \otimes \bar\lambda}(f)} 
	&= \frac{d_{\lambda}d_{\bar{\lambda}}}{\ffac{m}{n}} \sum_{k = c\sqrt{n}}^n \sum_{j=0}^k \sum_{(\mu|\rho) \in\classes_{n,k,j}} v_{(\mu|\rho)}  |\omega^{\lambda \otimes \bar\lambda}_{(\mu|\rho)}|^2\\
	&\leq d_{\lambda}d_{\bar{\lambda}} \sum_{k = c\sqrt{n}}^n \sum_{j=0}^k \sum_{(\mu|\rho) \in\classes_{n,k,j}} p_{(\mu|\rho)} \\
	&= d_{\lambda}d_{\bar{\lambda}} \sum_{k=c\sqrt{n}}^{n} \sum_{j=0}^k p_{k,j},
	\intertext{where the inequality holds by Proposition~\ref{prop:trivial}. Let $c = 4$. Proposition~\ref{prop:small} with $t = c\sqrt{n}$ implies}
	&\leq  d_\lambda d_{\bar{\lambda}} \cdot n^2 \cdot  \mathrm{O}(1/n^{\alpha c \sqrt{n}}).
	\intertext{By Proposition~\ref{prop:dims}, we have}
	&\leq \mathrm{O}\left( n^2 \cdot \frac{(n^{2 + \alpha})^\ell n^\ell}{n^{\alpha c \sqrt{n}}} \right)\\
	&\leq  \mathrm{O}\left(n^2 \cdot \frac{n^{3\ell}(n^\alpha)^{\ell}}{(n^\alpha)^{3\sqrt{n}}(n^\alpha)^{\sqrt{n}}} \right).
	\intertext{Since $\ell \leq \sqrt{n}$, setting $\alpha = 1+\varepsilon$ gives us} 
	&\leq  \mathrm{O}\left( n^2/n^{3\varepsilon \sqrt{n}} \right) = \mathrm{O}(1/\sqrt{n}).
	\end{align*}
	This proves (\ref{eq:main}), and thus the claim, which completes the proof of the main result.
\end{proof}
\noindent We have made no attempt to improve the dependency $m \geq n^{3+\varepsilon}$, and in fact, we believe that the claim above should be true for all $m \geq n(1+\varepsilon)$ such that $\varepsilon > 0$.

\section*{Acknowledgements}

The authors would like to thank Aleksandrs Belovs, Chris Godsil, and J\'{e}r\'{e}mie Roland for insightful discussions. The authors would also like to thank an anonymous referee for comments that substantially improved the readability of this paper.
This work was partially supported by the Singapore Ministry of Education and the National Research Foundation under grant R-710-000-012-135.

\bibliographystyle{alpha}
\bibliography{../../../master}

\end{document}